\documentclass[10pt]{article}

\usepackage{arxiv}

\usepackage[utf8]{inputenc}
\usepackage[T1]{fontenc}
\usepackage{amsmath,amsthm,amssymb}
\usepackage{mathtools}
\usepackage{newtxtext,newtxmath}
\usepackage{enumitem}
\usepackage{booktabs}
\usepackage{longtable}
\usepackage{tabularx}
\usepackage{graphicx}
\usepackage[section]{placeins}
\usepackage[hidelinks]{hyperref}
\usepackage{cleveref}
\graphicspath{{figures/}}
\hypersetup{
  pdftitle={Aggregate Disambiguation Systems},
  pdfauthor={José María Lago, Albert Castellana, Edgars Nemše},
  pdfsubject={Finite-panel inference for aggregate disambiguation},
  pdfkeywords={aggregate disambiguation, finite-sample certification, finite-population sampling, functional central limit theorem, Brownian bridge, non-adaptive corruption, LLM evaluation}
}

\newtheorem{theorem}{Theorem}[section]
\newtheorem{proposition}[theorem]{Proposition}
\newtheorem{corollary}[theorem]{Corollary}
\newtheorem{lemma}[theorem]{Lemma}
\theoremstyle{definition}
\newtheorem{definition}[theorem]{Definition}
\newtheorem{remark}[theorem]{Remark}

\newcommand{\E}{\mathbb{E}}
\newcommand{\Var}{\mathrm{Var}}
\newcommand{\Cov}{\mathrm{Cov}}
\newcommand{\Prob}{\mathbb{P}}
\newcommand{\ind}{\mathbb{1}}
\newcommand{\dTV}{d_{\mathrm{TV}}}
\newcommand{\Naturals}{\mathbb{N}}
\newcommand{\calT}{\mathcal{T}}
\newcommand{\calZ}{\mathcal{Z}}
\newcommand{\census}{\mathcal{P}}

\renewcommand{\shorttitle}{Aggregate Disambiguation Systems}
\renewcommand{\headeright}{Preprint}

\title{Aggregate Disambiguation Systems}
\author{José María Lago \quad Albert Castellana \quad Edgars Nemše\\[0.35em]
\small GenLayer Labs Research\\[-0.05em]
\small\texttt{jm@genlayerlabs.com} \quad
\texttt{albert@genlayerlabs.com} \quad
\texttt{edgars@genlayerlabs.com}}
\date{}

\begin{document}
\maketitle

\begin{abstract}
Natural-language tasks can elicit different verdicts from protocol-following evaluators that receive the same declared information.  We study aggregate disambiguation systems (ADSs).  Given a task and a candidate solution, each evaluator casts a binary vote on whether the solution should be accepted, and the system aggregates the votes of a finite panel.  The target is protocol reproducibility relative to an explicitly declared evaluator reference, not semantic truth.  We separate fixed finite censuses, probabilistic evaluator populations, and growing-census limits, since their endpoint laws and guarantees are not interchangeable.  In the population setting, we use finite samples to estimate how often a finite panel reaches the same decision as the declared evaluator population.  We provide a lower confidence bound on the fraction of candidate solutions for which the disagreement probability is at most a chosen tolerance.  The calculation accounts separately for sampling candidate solutions and sampling evaluators.  The construction permits arbitrary dependence among columns induced by shared evaluator rows and uses exact binomial intervals at the evaluator layer and an exact one-sided binomial inversion at the generator layer.  Simulations check the implementation against known population coverages and expose power limitations.  We illustrate the method with a pilot study in which 40 LLM evaluation runs each produced one binary vote for each of 50 manually selected cases.  The results apply only to these fixed cases.  Because the cases were not randomly sampled, the pilot does not support conclusions about a broader population of tasks.  Construction-underdetermined items illustrate why reproducibility must not be read as truth; no independent semantic adjudication is attempted.  We also give exact finite-census bounds and a finite-sample workload-profile certificate for explicitly non-adaptive targeted contamination.
\end{abstract}

\section{Introduction}\label{sec:introduction}

Many distributed decisions begin with a predicate that honest participants can
evaluate deterministically.  Natural-language tasks break that premise.  Two
protocol-following evaluators may receive the same problem, candidate answer,
instructions, and tools yet return different verdicts.  This occurs with human
reviewers, stochastic software, and large language models (LLMs).  The immediate
statistical question is therefore not whether one evaluator is infallible, but
whether a finite panel can reliably reproduce the decision of a declared
evaluator population.

We call a mechanism that generates a candidate resolution, evaluates it through
multiple independent episodes, and aggregates their verdicts an \emph{aggregate
disambiguation system} (ADS).  When a randomly sampled panel is itself the
decision mechanism and neither an external truth label nor a complete
population vote is available at run time, reproducibility is an operational
stability target: it asks whether the declared rule determines the action or
whether the particular panel draw does.  It remains separate from semantic
correctness.

The population in that definition is part of the estimand.  A finite catalogue
of frozen verdicts, a sequence of growing catalogues, and a probability
distribution from which fresh evaluator episodes are sampled answer different
questions and induce different sampling laws.  Treating them as interchangeable
is the source of several otherwise plausible but incorrect panel guarantees.

The paper studies three connected problems.  First, it identifies the exact
finite-panel error under each sampling regime.  Second, it asks how a finite
generator--evaluator experiment can certify the fraction of future resolutions
that a panel will classify reproducibly.  Third, it translates adversarial
participation into a sensitivity surface indexed by panel size and task clarity,
rather than a single context-free ``secure percentage.''

\subsection{Relation to prior work}

Several neighboring literatures overlap with one layer of ADS but use different
reference objects or observe different losses.  Consensus targets a common
state; collective-decision theories target accuracy under competence
assumptions or logical consistency across propositions; judge-panel and
measurement models target human labels, latent quality, or reliability; and
acceptance sampling, tolerance methods, and finite-sample risk control provide
related statistical tools.  The works reviewed below do not directly certify,
under a declared generator law, the mass of frozen resolutions for which a
fresh binary panel reproduces the decision of a predeclared evaluator
population.  We use that estimand as the common comparison axis below.

\paragraph{Consensus and non-deterministic replication.}
Byzantine agreement asks whether non-faulty processors can agree despite
arbitrary faults, while state-machine replication orders requests and executes
them from a common state~\cite{pease1980reaching,castro1999practical}.  For
non-deterministic operations, Cachin et al. filter or validate divergent
executions or provide shared randomness; Huang et al. separate agreement on
transaction order from agreement on the resulting state
~\cite{cachin2016nondeterminism,huang2024byzantine}.  These protocols engineer a
common accepted output or state despite faults or execution-level
non-determinism.

\paragraph{Collective judgment.}
The Condorcet jury theorem assumes an objectively correct alternative and
independent jurors who select it with probability greater than one half; under
those assumptions, majority accuracy increases with jury size
~\cite{condorcet1785essai}.  Peleg and Zamir characterize when analogous
conclusions survive under a general dependent voting law
~\cite{peleg2012extending}.  Judgment aggregation instead studies several
logically connected propositions: proposition-wise aggregation can produce an
inconsistent collective set, and quota thresholds determine when consistency
and completeness can coexist~\cite{list2002aggregating,dietrich2007judgment}.
Together these traditions distinguish truth-based accuracy, dependence, and
logical coherence as separate collective-decision criteria.

\paragraph{Aggregation during answer generation.}
Self-consistency samples multiple reasoning paths and aggregates their final
answers, while multiagent debate allows model instances to inspect and revise
one another's answers~\cite{wang2023selfconsistency,du2024multiagentdebate}.
In both, repeated calls help produce or select the answer rather than evaluate a
candidate frozen in advance.

\paragraph{LLM judges and judge panels.}
Zheng et al. compare LLM judgments with human preferences and document several
systematic biases; Verga et al. replace one large judge with a heterogeneous
panel; and Shi et al. study position bias and stability across repeated or
position-swapped evaluations
~\cite{zheng2023llmjudge,verga2024juries,shi2025judging}.  Qian et al. infer
latent rankings and judge effects without human labels, whereas Li shows in a
supervised setting that calibrating the full panel can outperform selecting
only its most accurate judges~\cite{qian2026trust,li2026calibrate}.  SCOPE uses
labelled calibration data to abstain while controlling error among
non-abstained pairwise judgments under exchangeability
~\cite{badshah2026scope}.  Finite-Calibration Panel Selection uses separate
selection, calibration, and validation data to choose a judge path, deployed
panel size, and aggregator family~\cite{zhu2026finitecalibration}.  Kohli
estimates that nine judges in one studied panel provide about two independent
votes' worth of information~\cite{kohli2026nine}.

\paragraph{Robust judge panels.}
Under Huber contamination, RoPoLL shows that mean aggregation of continuous or
vector-valued judge scores can be arbitrarily biased and replaces it with the
geometric median~\cite{acharya2026ropoll}.  It establishes a breakdown point of
one half, derives finite-sample upper and information-theoretic lower bounds
that agree on the parametric sampling rate but differ by a factor $\sqrt d$ in
the contamination floor, and extends the analysis to equicorrelated judge
outputs.  This is the closest robust-panel formulation, but
its target is latent-score estimation rather than reproduction under a fixed
binary rule.  Its breakdown and minimax results therefore do not transfer to
the margin-dependent target here, for which this paper makes no minimax optimality
claim.

\paragraph{Measurement and latent-answer models.}
Generalizability theory models persons, items, raters, and occasions as
measurement facets and estimates variance components to assess dependability
under different designs~\cite{cronbach1972dependability}.  Cultural consensus
theory estimates informant competence and a latent shared answer key, while
Dawid--Skene models latent response classes and observer-specific error rates
~\cite{romney1986culture,dawid1979maximum}.  These traditions quantify
reliability or infer latent truth and rater effects.

\paragraph{Finite populations, tolerance, and risk control.}
Risk-limiting election audits sample auditable records while controlling the
chance of confirming an incorrect reported outcome~\cite{stark2008conservative}.
Industrial acceptance-sampling plans use defect counts from a sample to accept
or reject a finite lot under a specified inspection design
~\cite{dodge1941sampling}.  Fleiss' kappa summarizes nominal agreement among
multiple raters relative to chance, and Krippendorff's alpha similarly compares
observed and chance-expected disagreement
~\cite{fleiss1971measuring,krippendorff2008systematic}.  Nonparametric tolerance
limits separate required population content from confidence in that content
~\cite{wilks1941tolerance,krishnamoorthy2009tolerance}.  Learn-then-Test uses
calibration data and multiple testing to select predictive-system parameters
with finite-sample risk control~\cite{angelopoulos2025learn}.

The ADS certificate combines these ingredients differently.  It derives exact
finite-census or superpopulation panel errors, certifies the otherwise
unobserved good-set membership from repeated evaluator rows, and then lower
bounds its mass under the generator law while selecting from a predeclared
panel-size grid.  Shared evaluator rows may make resolution columns dependent;
the certificate controls this either familywise or through an explicit
mass-of-failures allowance, without assuming cross-column independence.

\subsection{Contributions and organization}

The main contributions are:

\begin{enumerate}[leftmargin=*,itemsep=2pt,topsep=3pt]
  \item an operational estimand for panel reproducibility relative to a
  declared evaluator reference, requiring no semantic truth label, together
  with exact finite-census and superpopulation errors and explicit conditions
  and error envelopes for growing-census transfer;
  \item a nested finite-sample certificate in which repeated evaluator rows
  provide sound one-sided evidence for otherwise unobserved membership in the
  reproducible set and exact outer binomial inversion lower-bounds its generator mass,
  with familywise and mass-controlled variants that allow shared-row dependence
  and simultaneous selection from a predeclared panel-size grid; and
  \item exact worst-case results for non-adaptive finite-census corruption,
  together with a finite-sample certificate for the workload-level
  targeted-contamination sensitivity profile.
\end{enumerate}

Synthetic experiments check finite-sample validity and power, while the hash-pinned
real-data pilot is a worked implementation example under explicitly limited
reference objects; neither is treated as a separate methodological
contribution.

Section~\ref{sec:methodology} defines the estimands and certificate.
Section~\ref{sec:synthetic-results} reports synthetic results and adversarial
sensitivity.  Section~\ref{sec:applied-study} presents the LLM study, and
Section~\ref{sec:discussion} discusses the implications.  Complete definitions,
proofs, functional limits, and reproducibility details are collected in the
appendices and supplementary material.

\section{Methodology}\label{sec:methodology}

\subsection{System and estimand}

Fix an encoded problem $X$ and a concrete candidate resolution $T$.  An
evaluator episode $Z$ includes the evaluator configuration, prompt, tools, and
private randomness.  It produces component verdicts, after which a predeclared
global rule $G$ maps the component vector to one binary vote
\[
Y(Z;X,T)\in\{0,1\}.
\]
The order matters: the system first computes one global vote per evaluator and
then aggregates those votes.  Aggregating component-wise majorities and applying
$G$ afterwards is generally a different mechanism.  Appendix~\ref{sec:model}
gives the measurable construction.

Let $\nu_X$ be the declared probability law for evaluator episodes on problem
$X$.  For a frozen resolution $T$, write
\[
\mu_X(T)=\mathbb E_{Z\sim\nu_X}[Y(Z;X,T)].
\]
Given an inclusive threshold $\tau$, the population decision is one when
$\mu_X(T)\geq\tau$ and zero otherwise.  The \emph{clarity}
\[
\gamma_X(T)=|\mu_X(T)-\tau|
\]
is the distance from the population mean to the decision boundary.  It is an
operational margin, not a probability that the resolution is semantically true.

\subsection{Three inference regimes}

The reference object determines the sampling law.  Table~\ref{tab:regimes}
summarizes the three regimes used throughout the paper.

\begin{table}[htbp]
\centering
\small
\begin{tabularx}{\textwidth}{@{}p{0.23\textwidth}X p{0.25\textwidth}@{}}
\toprule
Reference object & Panel design and target & Exact endpoint law \\
\midrule
Frozen finite census & Uniform sampling without replacement; reproduce the
fixed census decision & Hypergeometric \\
Probabilistic population & Independent evaluator episodes drawn from a declared
law; reproduce its threshold decision & Binomial \\
Growing finite censuses & A deterministic approximation to a limiting
population; requires a stated convergence assumption or rate & Hypergeometric at
each finite stage \\
\bottomrule
\end{tabularx}
\caption{The three statistical regimes.  Similar numerical values do not make
their probability statements interchangeable.}
\label{tab:regimes}
\end{table}

For a census of $M$ frozen binary votes with $C$ positive entries, a uniform
panel of size $K$ has positive count
\[
H_{M,K}\sim\operatorname{Hypergeometric}(M,C,K).
\]
The exact probability that its threshold decision differs from the census
decision is therefore a hypergeometric tail.  No independence assumption is
made about the frozen votes.  Under the population law $\nu_X$, by contrast, a
fresh $K$-episode panel has count
\[
B_{X,K}(T)\sim\operatorname{Binomial}(K,\mu_X(T)),
\]
and its error relative to the population decision is the corresponding binomial
tail.  Full formulas and finite-population concentration bounds appear in
Appendix~\ref{sec:finite-population}.

When a deterministic census is asserted to approximate an ideal evaluator
population, that assertion needs an approximation envelope.  Under the margin
and variance conditions of Proposition~\ref{prop:finite-to-ideal-rate}, the
finite-census and ideal fixed-$K$ error probabilities differ by at most
$(K-1)/(M-1)+K b_{X,M}(T)$; without such an envelope, only the census-relative
statement is identified.

A uniform random ordering couples all panel sizes in a nested escalation.  The
centered superpopulation path converges to Brownian motion, whereas the path
through a fixed census is pinned at its known endpoint and converges to a
Brownian bridge.  These classical limits yield a closed late-reversal
probability under local clarity; the statements and proofs are in
Supplement~\ref{supp:functional-limits}.  Exact tails remain the basis of every
finite-sample certificate reported below.

\subsection{Resolution-level and workload-level error}

Let a generator configuration $u$ be drawn from a declared law $\pi$, and let
$T_u$ be its frozen resolution.  For panel size $K$, define $e_{X,K}(T_u)$ as the
exact binomial probability that a fresh evaluator panel disagrees with the
population decision for that resolution.  For a pointwise error tolerance
$\delta$, the population coverage is
\begin{equation}\label{eq:main-coverage}
R_{X;K,\delta}
=
\Pr_{u\sim\pi}\{e_{X,K}(T_u)\leq\delta\}.
\end{equation}
We call the problem resolvable at panel size $K$ if this coverage is at least
$1-\beta$, where $\beta$ is the permitted generator mass of unresolved
resolutions.  The two tolerances have different meanings: $\delta$ controls
panel error within a resolved unit, while $\beta$ controls how often a generated
unit may fail that pointwise requirement.

\subsection{Finite-matrix certificate}

The calibration experiment draws $A$ generator configurations and $M$ evaluator
rows independently from their declared laws.  Each row evaluates every frozen
resolution, producing an $M\times A$ binary matrix.  A row is the complete
random element for that replication: it may share an evaluator configuration,
evidence state, or coupled cell randomness across all columns.  Thus dependence
within a row, and hence across columns, may be arbitrary.  Rows themselves must
be independent and identically distributed, and for each fixed column its count
must have the declared binomial marginal.  Persistent random state shared across
rows must instead be conditioned on or modeled explicitly; it is outside this
theorem.

For each column $a$, an exact two-sided binomial interval $I_a^{\mathrm{mass}}$ estimates its unknown
acceptance mean.  The column certifies for $(K,\delta)$ only when the complete
interval lies on one side of the decision threshold and the worst-case binomial
tail over that interval is at most $\delta$.  Let
$\widehat R^{\mathrm{mass}}_{K,\delta}$ be the observed fraction of certified
columns and $S_K^{\mathrm{mass}}=A\widehat R^{\mathrm{mass}}_{K,\delta}$ their count.  Fix a
finite candidate grid $\mathcal K$, evaluator and generator failure budgets
$\eta_E$ and $\eta_G$, and an allowed mass $\xi_E$ of evaluator interval
failures.  Define
\begin{equation}\label{eq:main-lower-bound}
\underline R^{\mathrm{mass}}_{K,\delta}
=
\max\!\left\{0,
\ell_A\!\left(S_K^{\mathrm{mass}};\frac{\eta_G}{|\mathcal K|}\right)
-\xi_E\right\},
\end{equation}
where $\ell_A(s;\alpha)$ is the one-sided Clopper--Pearson lower limit after
$s$ successes in $A$ binomial trials at failure level $\alpha$.
Each evaluator interval is constructed at pointwise failure level
$\eta_E\xi_E$.  This level is independent of the number of generator columns.

\begin{theorem}[Mass-controlled operational certificate]
\label{thm:mass-operational-certificate}
Under the declared sampling design, with probability at least
$1-\eta_E-\eta_G$,
\[
R_{X;K,\delta}\geq\underline R^{\mathrm{mass}}_{K,\delta}
\qquad\text{simultaneously for every }K\in\mathcal K.
\]
Consequently, a panel size selected from the same matrix is certified whenever
its lower bound is at least $1-\beta$.
\end{theorem}

The proof uses exact marginal intervals, Tonelli's theorem and Markov's
inequality to control the population mass of interval failures, followed by
conditional exact binomial inversion over sampled generators.  It does not
assume that the columns are independent.  A closed-form Hoeffding alternative
replaces the exact lower bound by
$\max\{0,\widehat R^{\mathrm{mass}}_{K,\delta}-\xi_E-
\sqrt{\log(|\mathcal K|/\eta_G)/(2A)}\}$ and is valid but generally more
conservative.  Appendix~\ref{subsec:operational-certificate} contains the full
proof, a stronger familywise variant, finite-catalogue versions, and the
finite-to-population transfer conditions.

\subsection{Adversarial sensitivity models}

We distinguish two non-adaptive models.  In the \emph{fixed-share} model, each
sampled identity is Byzantine with probability $\alpha$ and votes against the
honest decision; an honest identity supports that decision with probability
$1/2+\gamma_H$.  The unconditional support probability is then
\begin{equation}\label{eq:main-fixed-share}
p_{\mathrm{net}}=(1-\alpha)(1/2+\gamma_H).
\end{equation}
In the stronger \emph{targeted-contamination} model, the adversary can replace a
fraction $\alpha$ of entries specifically among those supporting the honest
decision, reducing an honest-side mean $1/2+\gamma$ to $1/2+\gamma-\alpha$.
Both attack probabilities are exact binomial tails once $K$, $\alpha$, and the
relevant clarity are fixed.  Appendix~\ref{sec:corruption} gives the complete
finite-census results and a simultaneous finite-sample certificate for the
workload-level targeted-contamination profile.

\section{Synthetic Results}\label{sec:synthetic-results}

The synthetic experiments serve two purposes: they verify the implementation of
the finite-sample certificate against known population quantities, and they show
which combinations of panel size, clarity, and adversarial participation can be
distinguished at practical sample sizes.  They do not estimate an LLM
population.

\subsection{Exact finite-panel behavior}

Figure~\ref{fig:finite-error} evaluates the exact hypergeometric disagreement
probability for a frozen census of 1,500 votes.  Error falls quickly with panel
size when the census mean is far from the threshold and slowly when it is close.
The binomial approximation is accurate for small sampling fractions but does
not incorporate the finite-population correction and does not become exact at
the full census.

\begin{figure}[htbp]
\centering
\includegraphics[alt={Log-scale panel disagreement versus panel size for several census clarity values, comparing exact hypergeometric and binomial curves.},width=0.90\textwidth]{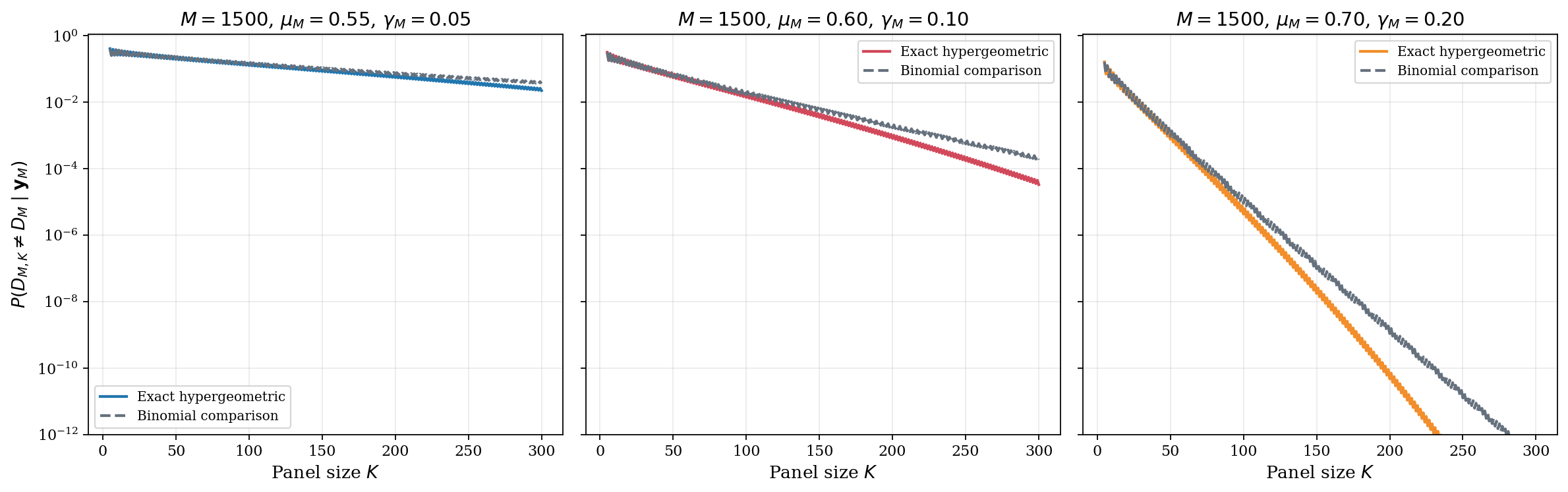}
\caption{Exact panel--census disagreement for a census of 1,500 votes under an
inclusive majority rule.  Curves are indexed by the census clarity.}
\label{fig:finite-error}
\end{figure}

This calculation illustrates why panel size alone is not a security or quality
parameter.  Away from the threshold, the required $K$ scales with the inverse
square of clarity.  At the superpopulation boundary, panel decisions do not
stabilize as $K\to\infty$; in a finite census, however, the without-replacement
path is pinned and becomes exact at $K=M$.

\subsection{Finite-sample validity and power}

We simulated generator units with latent acceptance means
\[
(0.30,0.40,0.46,0.49,0.51,0.54,0.60,0.70)
\]
and weights
$(0.15,0.20,0.05,0.10,0.10,0.05,0.20,0.15)$.  The threshold was one half,
the pointwise error target was $0.05$, and the required population coverage was
$0.60$.  Candidate panel sizes ranged from 21 to 301.  We compared independent
columns with a shared-row mixture in which, with probability $0.9$, every column
uses the same row-level uniform draw.  The latter construction preserves the
exact binomial marginal of every column while inducing strong dependence across
columns.

Each configuration was repeated 2,000 times.  A violation occurs if a
simultaneous lower bound exceeds the known population coverage at any candidate
panel size.  Table~\ref{tab:operational-monte-carlo} reports no violations in
8,000 experiments for either exact outer binomial inversion or the Hoeffding
benchmark.  This is an implementation check rather than a proof; the proof is
Theorem~\ref{thm:mass-operational-certificate}.  The comparison also shows the
distinction between validity and power.  The baseline experiment is valid but
usually cannot certify a true coverage of $0.70$ against a target of $0.60$,
whereas the powered design usually can.

\begin{table}[htbp]
\centering
\small
\begin{tabular}{@{}lrrrrrr@{}}
\toprule
Design & $A$ & $M$ & $\rho$ & Violations & Target exact / H & Mean exact / H \\
\midrule
Baseline & 500  & 1500 & 0.0 & 0 / 0 & 0.00\% / 0.00\% & 0.512 / 0.497 \\
Baseline & 500  & 1500 & 0.9 & 0 / 0 & 10.25\% / 2.85\% & 0.512 / 0.496 \\
Powered  & 2000 & 3000 & 0.0 & 0 / 0 & 96.15\% / 83.75\% & 0.619 / 0.610 \\
Powered  & 2000 & 3000 & 0.9 & 0 / 0 & 95.45\% / 85.65\% & 0.619 / 0.610 \\
\bottomrule
\end{tabular}
\caption{Monte Carlo finite-sample validity and power.  H denotes the Hoeffding
benchmark; violation entries give exact / H counts out of 2,000.  The final
column is evaluated at panel size 301.}
\label{tab:operational-monte-carlo}
\end{table}

\begin{figure}[htbp]
\centering
\includegraphics[alt={Monte Carlo finite-sample validity and power under independent columns and strong shared-row dependence.},width=0.96\textwidth]{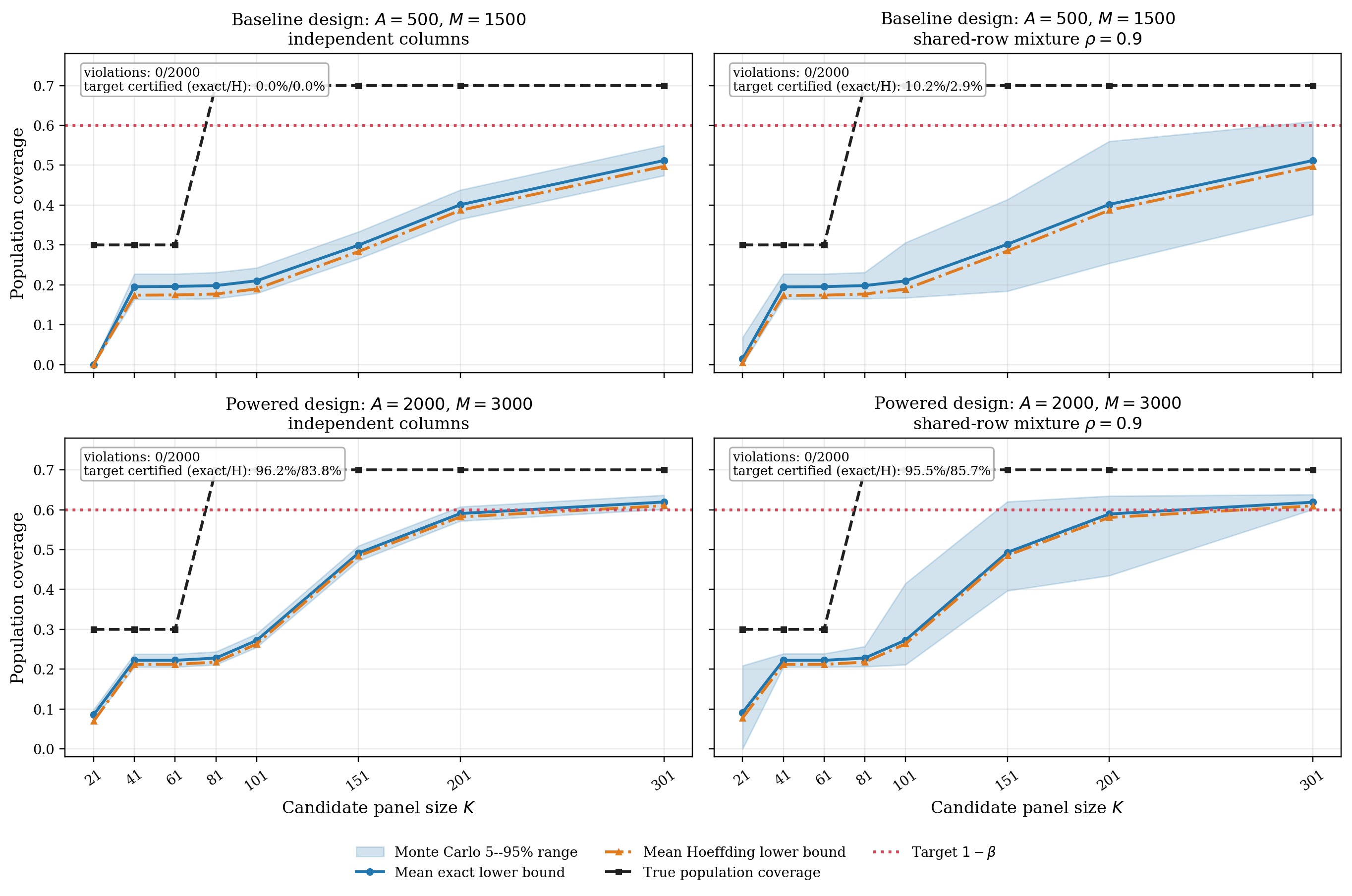}
\caption{Known population coverage (black), mean exact lower bound (blue),
its 5th--95th percentile range (band), mean Hoeffding bound (orange), and
target coverage (red).  Shared rows increase dispersion without invalidating
the marginal certificate.  Figure annotations use one decimal place;
Table~\ref{tab:operational-monte-carlo} reports the empirical percentages to
two decimal places.}
\label{fig:operational-monte-carlo}
\end{figure}

\FloatBarrier
\subsection{Panel size, clarity, and adversarial participation}

We next evaluated the fixed-share model on the nominal GenLayer panel-size grid.
Figure~\ref{fig:panel-clarity-adversary} plots the minimum clarity required within
the honest subpopulation to keep the exact outcome-manipulation probability at
or below one percent.  Larger panels reduce finite-sampling uncertainty, but
they cannot remove a population-level adversarial shift.  For a fixed Byzantine
share $\alpha$, the curves approach
\[
\gamma_H=\frac{\alpha}{2(1-\alpha)}.
\]
Thus the answer to ``what percentage is safe?'' is a curve conditional on task
clarity, panel size, attack model, and error target.

\begin{figure}[htbp]
\centering
\includegraphics[alt={Minimum honest-subpopulation clarity versus panel size, with one curve for each fixed Byzantine network share.},width=0.94\textwidth]{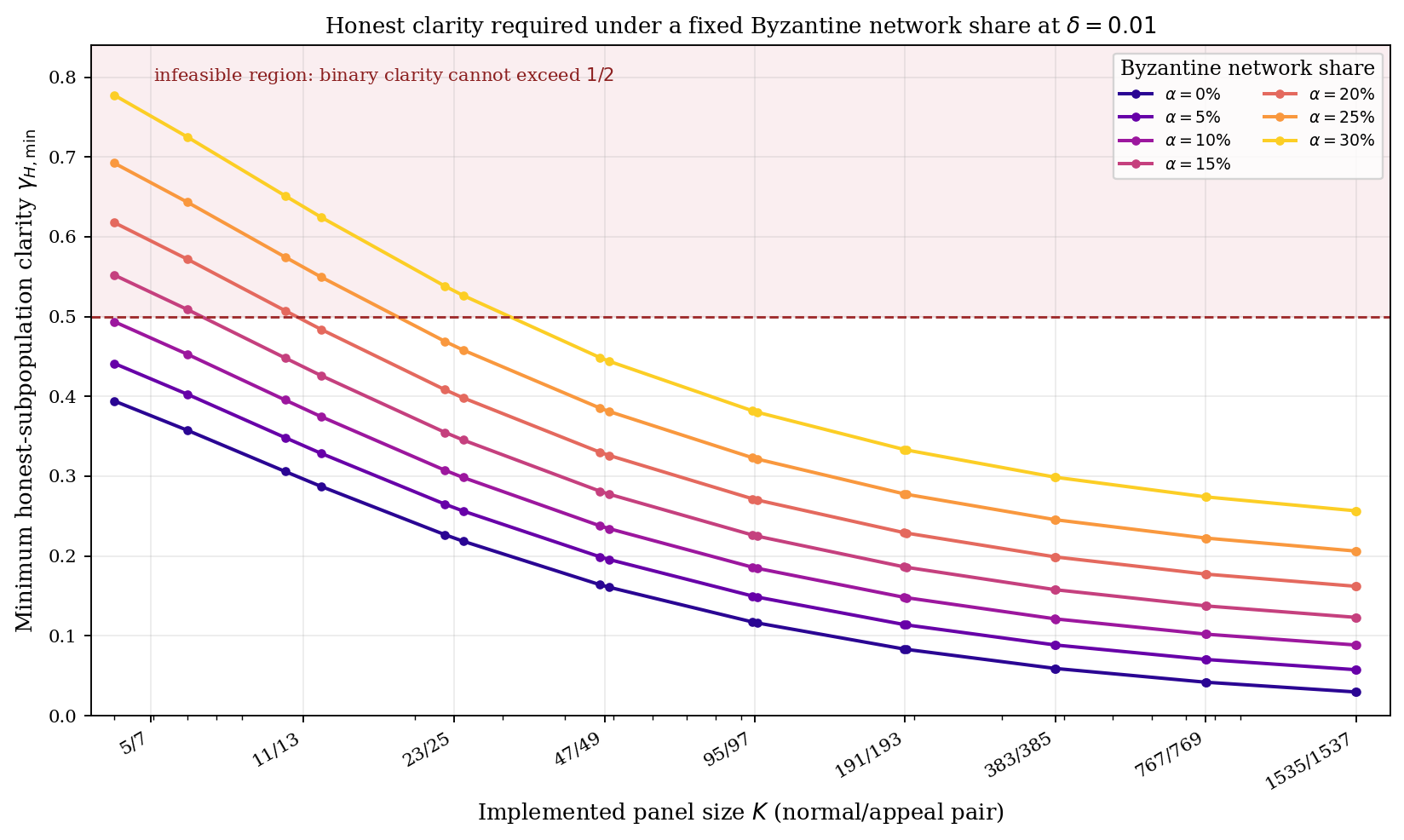}
\caption{Minimum honest-subpopulation clarity required for pointwise attack risk
at most one percent under the fixed-share model.  Values above one half are
infeasible for a binary threshold.  This benchmark is not a deployed-network
security guarantee.}
\label{fig:panel-clarity-adversary}
\end{figure}

Fixing the largest nominal panel, Figure~\ref{fig:attack-phase-diagram} compares
the fixed-share and targeted-contamination models over the complete
clarity--adversary plane.  The diagonal in the targeted model is not a claim that
network ownership and clarity are the same quantity.  It follows from the
stronger assumption that every unit of the corruption budget can be spent on an
entry supporting the honest decision.  Under random network ownership the
frontier is curved because Byzantine identities displace a mixture of latent
honest votes.

\begin{figure}[htbp]
\centering
\includegraphics[alt={Side-by-side attack phase diagrams at panel size 1537 for a fixed Byzantine network share and targeted census contamination.},width=0.80\textwidth]{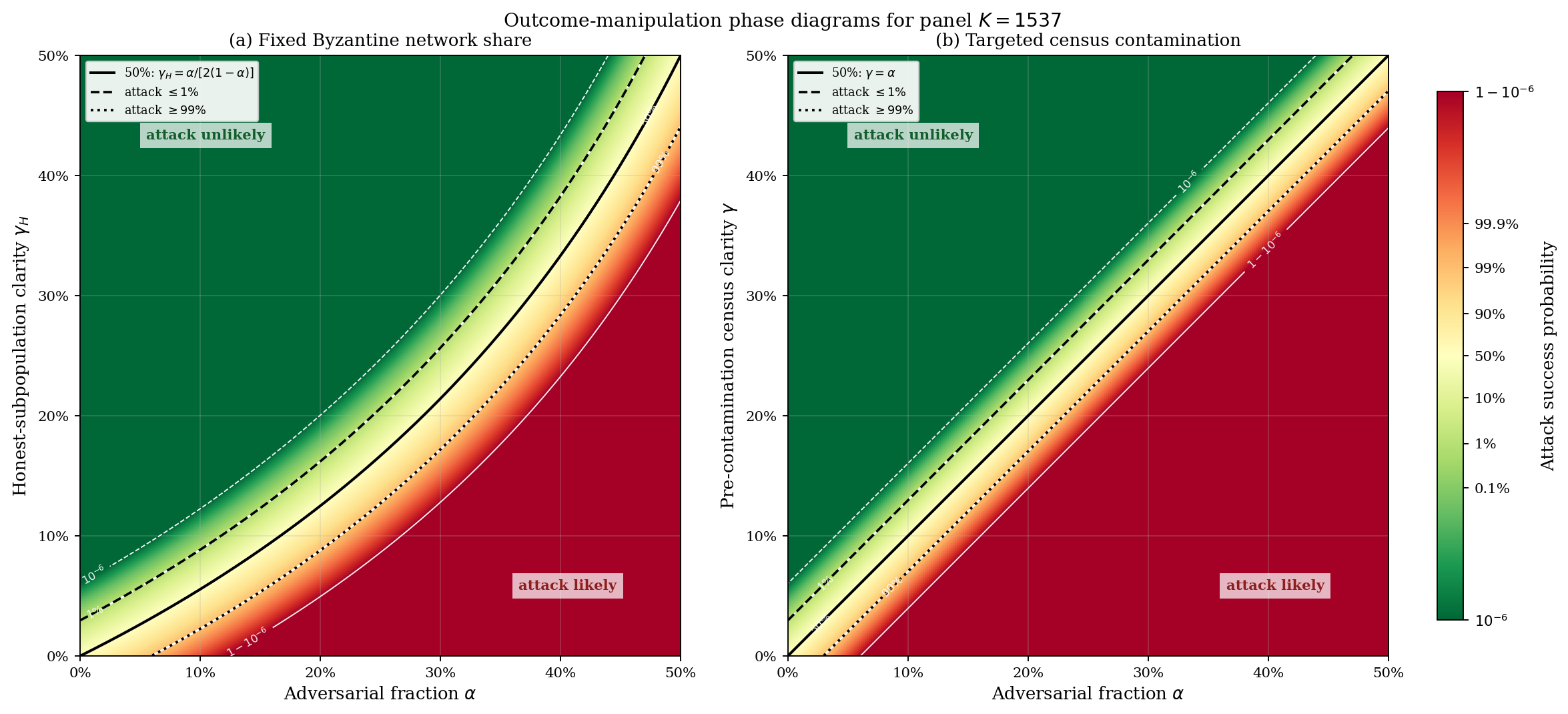}
\caption{Exact outcome-manipulation probability at panel size 1,537.  Left:
fixed Byzantine network share.  Right: targeted census contamination.  Solid
black lines mark 50\% attack probability; dashed and dotted lines mark 1\% and
99\%.  This benchmark is not a deployed-network security guarantee.}
\label{fig:attack-phase-diagram}
\end{figure}

The nominal grid is used only as an application-anchored sensitivity axis.  It
does not reproduce the deployed stake-weighted sampler, validator reuse,
path-dependent appeal feasibility, or adaptive corruption.  Those mechanisms
must be combined with the task-specific curves before making a network-security
claim.

\section{Applied Case: A Real-Data LLM Pilot}\label{sec:applied-study}

\subsection{Study design}

We ran a hash-pinned diagnostic pilot using real calls through an LLM router.
The 50 frozen workload units were manually stratified into clear acceptance,
clear rejection, benign ambiguity, and input-presentation stress.  They form a
fully enumerated seed catalogue, not an i.i.d.\ sample from an identified
workload--generator law.  Labels describe how the units were constructed; they
were not independently adjudicated truth labels.

Each prompt requested one binary global verdict for one frozen decision unit.
Thus the pilot exercises the $N=1$ computational path of the ADS formalism,
where $G$ is the identity.  It does not empirically exercise a multicomponent
verdict vector followed by a nontrivial global map.  Evaluator rows were sampled
with replacement from a predeclared equal-weight catalogue of routed
configurations and reused across all units.  This induces cross-unit dependence
allowed by the shared-row theorem, provided rows are i.i.d.\ from the declared
law.  Distinct seeds cannot exclude persistent dependence caused by provider
state, routing incidents, tools, or caches; those possibilities remain an
explicit limitation.

\subsection{Results under three reference objects}

The same binary matrix supports different statements under different reference
objects.  We report them separately.

\paragraph{Fully enumerated unit catalogue.}
Treating the 50 units and their declared weights as the complete outer
catalogue removes generator-sampling uncertainty.  At the retained audit value
$K=47$, familywise evaluator intervals certify 45 of 50 units.
Corollary~\ref{cor:finite-generator-catalogue} therefore gives lower coverage
$0.900$ under uniform catalogue weights and $0.8775$ under the declared weights,
with confidence $0.975$ under the declared i.i.d.\ evaluator-row model.  These
are formal statements about this fixed catalogue, not about future workload
units.  The nonuniform law was frozen before evaluator calls: if unit $i$
belongs to family $f$, then $r_i=w_f/n_f$, where the 13 family masses $w_f$ and
family sizes $n_f$ are reported in Table~\ref{tab:catalogue-weights}.  They are
a design prior for this finite catalogue, not frequencies estimated from the
observed votes or from a deployment population.

\paragraph{Frozen evaluator census.}
Conditioning instead on the 40 realized global verdicts per unit eliminates
evaluator-superpopulation inference.  Exact hypergeometric calculations show
that 47 of 50 units have
$K_{\mathrm{stable}}\leq7$ at $\delta=0.01$.  The remaining units are
\texttt{mkt-03}, \texttt{ins-03}, and \texttt{air-03}, with stable panel sizes
26, 34, and 39, respectively.  Here panels are sampled without replacement and
necessarily satisfy $K\leq40$.

Thus the same matrix gives 47 of 50 exact frozen-census stability results at
$K=7$, but zero evaluator-superpopulation interval certificates at that panel
size.  The contrast is not a contradiction: it is the empirical reason the
finite-census and superpopulation regimes cannot be interchanged.

\paragraph{Outer-sampling diagnostic.}
If the 50 units had instead been sampled i.i.d.\ from a declared workload law,
the mass-controlled theorem would apply.  At $K=47$, 46 units pass its pointwise
interval rule.  With $\eta_G=0.025$, $|\mathcal K|=18$, and $\xi_E=0.05$, exact
outer binomial inversion gives
$\ell_{50}(46;0.025/18)-0.05=0.6912$; the valid but more conservative Hoeffding
form gives $46/50-0.05-\sqrt{\log(18/0.025)/100}=0.6135$.  At $K=23$, 44 units
pass and the exact value is $\ell_{50}(44;0.025/18)-0.05=0.6372$, which first
crosses the pilot target $0.60$.  Under the hypothetical i.i.d.\ outer design,
these would be simultaneous lower bounds at confidence
$1-\eta_E-\eta_G=0.95$.  The division by 18 retains the multiplicity cost of
the complete predeclared grid.  The original Hoeffding analysis first crossed
at $K=47$; we retain that value as the audit reference rather than selecting a
smaller value after inspecting the matrix.  Because the actual catalogue was
manually stratified, none of these values carries a workload-population
confidence interpretation here.

\begin{figure}[htbp]
\centering
\includegraphics[alt={Observed clarity by construction stratum and fixed-catalogue and outer-sampling diagnostics over panel size.},width=0.95\textwidth]{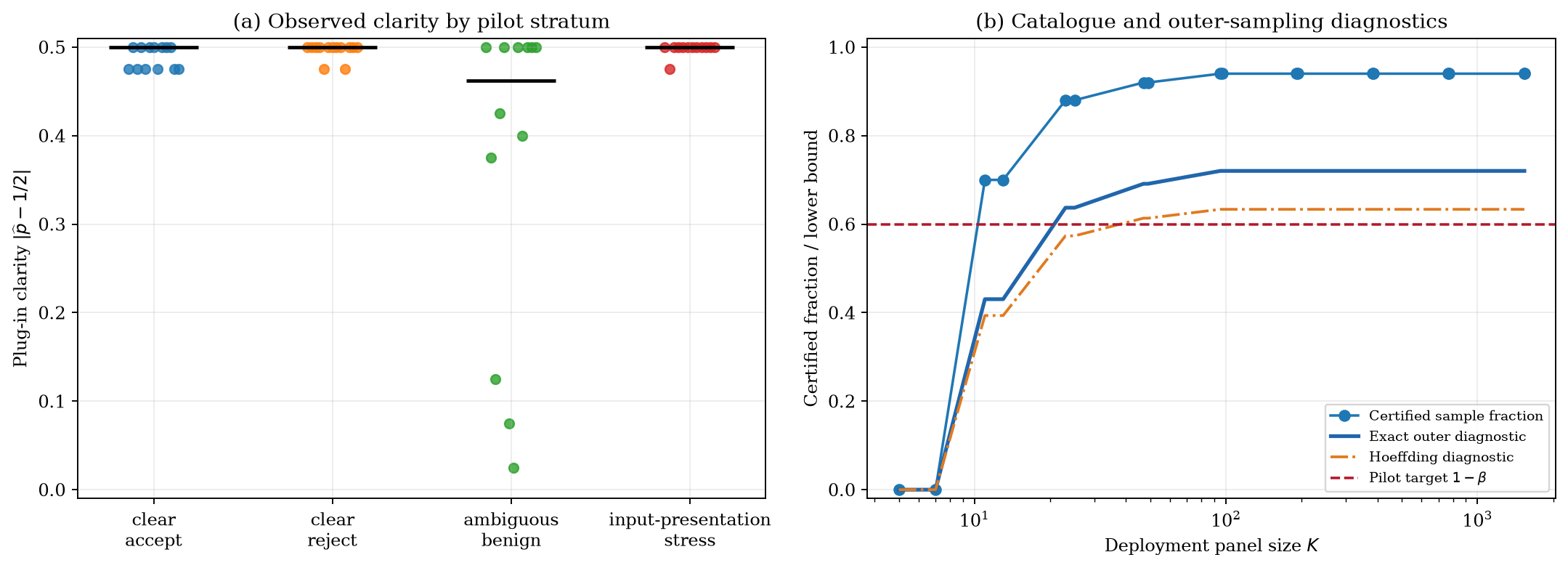}
\caption{Left: observed plug-in clarity by construction stratum.  Right:
certified catalogue fraction and two outer-sampling diagnostics.  Exact
binomial inversion crosses the pilot target at $K=23$; the original Hoeffding
analysis crosses at $K=47$.  Neither outer curve is a workload-population
estimate for this stratified catalogue.}
\label{fig:real-llm-pilot}
\end{figure}

\FloatBarrier
\subsection{Diagnostics and interpretation}

The empirical majority matched the construction label on all 37 units designed
to have a determinate outcome.  This is construction agreement, not an
externally validated accuracy estimate.  Three evaluator intervals crossed the
decision threshold: \texttt{mkt-03}, \texttt{ins-03}, and \texttt{air-03}.
Equal weighting of the 18 realized configurations changed no unit decision;
removing one realized configuration at a time changed only
\texttt{air-03}.  Mapping every non-success to acceptance also changed only
\texttt{air-03}, while both that policy and a successes-only analysis changed
the certification status of \texttt{bug-03}.  These reweighting and failure
analyses are descriptive: the archived rows were not sampled to certify the
alternative evaluator laws.

Applying the non-adaptive census-contamination model of
Theorem~\ref{thm:corrupt-exact} to the same ledger shows that, at $K=7$, the
number of catalogue units meeting $\delta=0.01$ is 47 with no flips, 46 after
two worst-direction flips per 40-vote column, and 44 after four.  These are
conditional frozen-census sensitivities, not guarantees against adaptive
corruption or a deployed network; Table~\ref{tab:pilot-census-corruption}
reports the fuller calculation.

Nine of the 13 units constructed to be underdetermined nevertheless passed the
mass-controlled pointwise interval rule at the audit value $K=47$.  This
establishes neither
that they are semantically underdetermined nor that the resulting convention is
correct; no independent semantic adjudication was performed.  It does
illustrate the estimand's limitation: reproducibility relative to a declared
population can coexist with construction-level ambiguity, shared convention,
or shared bias.  External correctness requires a separate labelled study.

\section{Discussion}\label{sec:discussion}

The results turn panel reproducibility into an estimable quantity.  A frozen
census supports exact design-based claims; a probabilistic evaluator population
supports fresh-panel claims; and a growing deterministic catalogue supports a
population claim only when its approximation is justified.  The operational
certificate then combines evaluator uncertainty and generator sampling without
pooling resolution-specific margins.  This separation is especially important
for LLM systems, where shared model families can produce strong cross-task
dependence and strong common bias at the same time.

The synthetic and applied results make two practical points.  First, failure to
certify can mean genuine ambiguity or inadequate experimental power; the
baseline Monte Carlo design is valid but underpowered.  Second, successful
certification establishes reproducibility rather than truth.  The LLM study
found both near-perfect coordination on constructed clear cases and high
coordination on cases constructed to be underdetermined.  Because those
construction labels were not independently adjudicated, the pilot does not
estimate a relation between semantic underdetermination and coordination.  A
deployment should report the panel-reproduction guarantee alongside independent
semantic metrics whenever adjudicated truth is available.

Adversarial tolerance is likewise not a universal percentage.  It depends on
the task-specific clarity distribution, panel size, sampling mechanism, and
attacker model.  Larger panels suppress random committee variation but amplify
the decision of whichever population remains after contamination.  The phase
diagrams should therefore be combined with a pinned deployment population and
inclusion law, rather than quoted as standalone network-security guarantees.
When workload units and honest evaluator rows follow the declared sampling
design, Corollary~\ref{cor:contamination-profile-certificate} supplies a
finite-sample lower bound for the targeted-contamination profile; without that
design, the same curves remain sensitivity analyses only.

The main limitations also define the immediate research program.  The present
certificate assumes a predeclared finite grid, i.i.d.\ rows, and correct
binomial marginals.  Reweighting a matrix drawn under one evaluator law is only
descriptive unless an appropriate alternative-law design and inference are
supplied.  The real-data pilot contains only 50 manually stratified units,
exercises the $N=1$ identity map rather than a nontrivial multicomponent $G$,
and has no independent semantic adjudication.  More generally,
time-uniform confidence sequences would permit optional escalation, while
clustered or latent-factor models would address shared provider state.  Real
deployments additionally require stake-weighted and path-dependent inclusion
laws, adaptive adversaries, temporal replication, and independently adjudicated
workloads.  Finally, panel cost, latency, appeals, and bonds should be optimized
jointly with the statistical error budget.  These extensions change the
deployment layer, but not the central distinction between reproducibility,
population choice, and semantic correctness.

\section*{Institutional Disclosure}
This work was conducted within GenLayer Labs Research.  GenLayer Labs develops GenLayer, whose nominal panel-size ladder is used as an application-anchored sensitivity grid in Section~\ref{sec:synthetic-results}.  The statistical results do not depend on that implementation.  No production deployment outcomes or user data are analyzed.

\bibliographystyle{plain}
\bibliography{references}

\clearpage
\appendix
\renewcommand{\headeright}{Appendices}
\pdfbookmark[0]{Appendices: Proofs and Technical Results}{part1-appendices}
\documentdivision{Appendices: Proofs and Technical Results}
\section{Notation Summary}\label{app:notation}

\small
\begin{longtable}{@{}p{0.29\textwidth}p{0.66\textwidth}@{}}
\toprule
\textbf{Symbol} & \textbf{Meaning}\\
\midrule
\endfirsthead
\toprule
\textbf{Symbol} & \textbf{Meaning}\\
\midrule
\endhead
$\Sigma^*$ & Countable set of finite encoded words\\
$X$ & Fixed problem with $N$ components\\
$T$ & Fixed concrete resolution of $X$\\
$G$ & Common global map $\{0,1\}^N\to\{0,1\}$\\
$\Theta$ & Evaluator-configuration space\\
$\calZ=\Theta\times[0,1]$ & Complete evaluator-episode space\\
$\nu_X$ & Declared task-specific evaluator measure\\
$V_j(z;X,T)$ & Component verdict\\
$Y(z;X,T)$ & Global evaluator verdict\\
$\mathbf Y$ & Joint vote element in $\{0,1\}^{\Naturals}$\\
$\mathcal U,\pi$ & Generator-configuration space and declared law\\
$M$ & Evaluator-census or evaluator-row sample size\\
$A$ & Generator-column sample size\\
$K$ & Panel size (finite-census regime: $K\leq M$)\\
$C_M$ & Number of positive census verdicts\\
$\mu_M=C_M/M$ & Census acceptance proportion\\
$\mu_X(T)$ & Limiting or superpopulation acceptance mean\\
$\tau$ & Inclusive decision threshold\\
$q_K=\lceil\tau K\rceil$ & Integer positive-vote quota\\
$D_{M,K}$ & Panel decision\\
$D_M$ & Finite census decision\\
$D_{X,\infty}$ & Limiting population decision\\
$\gamma_M$ & Census clarity $|\mu_M-\tau|$\\
$\gamma_X$ & Limiting clarity $|\mu_X-\tau|$\\
$H_{M,K}$ & Positive panel count\\
$e_{M,K}(T)$ & Exact panel--census disagreement probability\\
$e_{X,K}(T)$ & Exact panel--limiting-population disagreement probability\\
$B^\circ$ & Standard Brownian bridge\\
$W$ & Standard Brownian motion\\
$R_{X;K,\delta}$ & Population resolvability coverage\\
$R_{K,\delta}^{(A,M)}$ & Empirical resolvability coverage\\
$p(u),e_K(u)$ & Workload-unit acceptance mean and exact panel-error tail\\
$R_{\Pi;K,\delta}$ & Coverage under a declared workload--generator law $\Pi$\\
$e_K^{\mathrm{acc}}(p),e_K^{\mathrm{rej}}(p)$ & Directed binomial panel-error tails\\
$\operatorname{Cert}_{K,\delta}(I)$ & Sound interval certificate for one column\\
$\widehat R_{K,\delta}^{\mathrm{cert}},\widehat R_{K,\delta}^{\mathrm{mass}}$ & Fractions certified under the familywise and mass-controlled constructions\\
$S_K^{\mathrm{cert}},S_K^{\mathrm{mass}}$ & Counts of certified generator columns under the familywise and mass-controlled interval constructions\\
$\ell_A(s;\alpha)$ & One-sided exact binomial lower limit after $s$ successes in $A$ trials\\
$L_K^{G,\mathrm{cert}},L_K^{G,\mathrm{mass}}$ & Outer binomial limits before any mass charge\\
$\underline R_{K,\delta}^{\mathrm{cert}},\underline R_{K,\delta}^{\mathrm{mass}}$ & Familywise and mass-controlled simultaneous lower coverage bounds\\
$Q_K(\mathbf Z)$ & Conditional generator mass of positive interval certificates\\
$F(\mathbf Z)$ & Generator mass of evaluator intervals that miss their means\\
$\beta$ & Maximum unresolved fraction\\
$\eta_E,\eta_G$ & Evaluator- and generator-level calibration failure budgets\\
$\xi_E$ & Evaluator-interval failure mass charged to coverage\\
$I_a=[L_a,U_a]$ & Confidence interval for the evaluator mean of column $a$\\
$\mathcal K$ & Predeclared finite set of candidate panel sizes\\
$b$ & Non-adaptive census-corruption budget\\
$\alpha$ & Adversarial fraction or asymptotic contamination level\\
$P_K^{\mathrm{cap}}(\alpha)$ & Strict-majority committee-capture probability\\
$P_K^{\mathrm{net}}(\alpha;\gamma_H)$ & Fixed-share network attack probability\\
$P_K^{\mathrm{tar}}(\alpha;\gamma)$ & Targeted-contamination attack probability\\
$r_{K,\delta}$ & Post-corruption clarity required for pointwise error $\delta$\\
$S_{K,\delta}(\alpha)$ & Workload non-adaptive corruption profile\\
$\widehat S_{K,\delta}^{\mathrm{mass}}(\alpha)$ & Observed mass-controlled targeted-certificate fraction\\
$\underline S_{K,\delta}(\alpha)$ & Simultaneous lower bound on the targeted-contamination profile\\
$\alpha^*_{K,\delta,\beta}$ & Largest contamination level whose workload profile is at least $1-\beta$\\
\bottomrule
\end{longtable}
\normalsize

\section{Formal ADS Model and Binary Process}\label{sec:model}

\subsection{Encoded Problems and Concrete Resolutions}

Let $\Sigma$ be a finite non-empty alphabet and let
\[
\Sigma^*=\bigcup_{n\geq0}\Sigma^n
\]
be the set of finite words, equipped with the discrete sigma-algebra $2^{\Sigma^*}$.  The encoding is syntactic: no algebraic or metric meaning is assigned to bytes or tokens.  In particular, the theory never averages textual outputs.

\begin{definition}[Problem]\label{def:problem}
A problem is
\[
X=((x_1,\ldots,x_N),\iota),
\]
where $N\in\Naturals$ with $N\geq1$, every $x_j\in\Sigma^*$ is a component specification, and $\iota\in\Sigma^*$ encodes the common evaluation instructions and permitted resources.
\end{definition}

\begin{definition}[Concrete Resolution]\label{def:resolution}
A concrete resolution of $X$ is a vector
\[
T=(w_1,\ldots,w_N)\in(\Sigma^*)^N.
\]
Once generated, $T$ is frozen and supplied identically to every evaluator in the experiment under study.
\end{definition}

The finiteness of $N$ concerns one concrete problem.  The population of possible resolutions, evaluator episodes, and sampled panels may all be infinite.

\subsection{Evaluator Episodes}

Let $(\Theta,\mathfrak G,\nu)$ be a probability space of evaluator configurations.  A configuration includes the model or human-evaluator type, fixed prompt, tool policy, and other declared design choices.  Let $([0,1],\mathcal B([0,1]),\lambda)$ supply private randomness.  A single evaluator episode belongs to
\[
(\calZ,\mathfrak A,\rho)
=
(\Theta\times[0,1],\mathfrak G\otimes\mathcal B([0,1]),\nu\otimes\lambda).
\]
The seed coordinate may encode a countable sequence of independent draws, including private tool observations, through the construction in Supplement~\ref{supp:seed}.  A genuinely shared random environment is not private randomness: it must be conditioned upon or added as a common coordinate, in which case unconditional evaluator votes need not be independent.

\begin{definition}[Interpreter Family]\label{def:interpreter-family}
An interpreter family is a map
\[
\mathcal I:\Theta\times\Sigma^*\times[0,1]\longrightarrow\Sigma^*
\]
such that $(\theta,\omega)\mapsto\mathcal I(\theta,u,\omega)$ is $\mathfrak A$-measurable for every fixed $u\in\Sigma^*$.
\end{definition}

Let $Q:(\Sigma^*)^3\to\Sigma^*$ be a deterministic prompt constructor and let $\mathcal A_+\subseteq\Sigma^*$ be a measurable set of explicit affirmative responses.  Malformed and non-affirmative responses count as zero.  Let $\omega\mapsto(\omega^{(1)},\omega^{(2)},\ldots)$ denote the measurable splitting of one uniform seed into countably many independent uniform seeds.

\begin{definition}[Component Verdict Vector]\label{def:component-vector}
For $z=(\theta,\omega)\in\calZ$, problem $X$, and resolution $T$, define
\[
V_j(z;X,T)
=
\ind\!\left[
\mathcal I\bigl(\theta,Q(\iota,x_j,w_j),\omega^{(j)}\bigr)\in\mathcal A_+
\right],
\qquad j=1,\ldots,N,
\]
and
\[
\mathbf V(z;X,T)=(V_1(z;X,T),\ldots,V_N(z;X,T))\in\{0,1\}^N.
\]
No independence is assumed among the $N$ coordinates of one evaluator's vector.
\end{definition}

An evaluator may internally construct a private solution or consult external evidence before returning a component verdict.  Such behavior is part of the measurable map $\mathcal I$ and the episode $z$; the probability theory does not require a particular internal reasoning procedure.

\begin{definition}[Global Classification Rule and Verdict]\label{def:global-verdict}
A global classification rule is a deterministic map
\[
G:\{0,1\}^N\longrightarrow\{0,1\}
\]
fixed before evaluations are observed.  The global verdict is
\[
Y(z;X,T)=G\bigl(\mathbf V(z;X,T)\bigr).
\]
\end{definition}

For the strict examination rule,
\[
G(v_1,\ldots,v_N)=\prod_{j=1}^N v_j;
\]
the resolution passes exactly when no component fails.

\begin{proposition}[Measurability]\label{prop:verdict-measurable}
For fixed $(X,T)$, the global verdict
\[
Y(\,\cdot\,;X,T):(\calZ,\mathfrak A)\longrightarrow(\{0,1\},2^{\{0,1\}})
\]
is measurable.
\end{proposition}

\begin{proof}
For each $j$, the output of $\mathcal I$ at the fixed prompt $Q(\iota,x_j,w_j)$ and measurable split seed $\omega^{(j)}$ is measurable.  Taking the inverse image of $\mathcal A_+$ and its indicator preserves measurability.  The vector of finitely many measurable coordinates is measurable into the finite discrete product, and every map $G$ on that product is measurable.
\end{proof}

\begin{remark}[Order of Operations]\label{rem:global-first}
The ADS first computes one global vote
\[
Y_i=G(V_{i,1},\ldots,V_{i,N})
\]
for each evaluator and only then aggregates the $Y_i$.  Applying $G$ to component-wise population majorities is a different rule in general and is not analyzed here.
\end{remark}

\subsection{The Countable Binary Random Element}\label{sec:binary-process}

Fix $(X,T)$.  On the countable product probability space
\[
(\calZ^{\Naturals},\mathfrak A^{\otimes\Naturals},\rho^{\otimes\Naturals}),
\]
let $Z_i$ be the $i$th coordinate projection and set
\[
Y_i=Y(Z_i;X,T)\in\{0,1\}.
\]

\begin{definition}[Joint Vote Element]\label{def:joint-vote}
The joint vote element is
\[
\mathbf Y_{X,T}
=(Y_1,Y_2,\ldots):
\calZ^{\Naturals}\longrightarrow\{0,1\}^{\Naturals},
\]
where the codomain carries the product sigma-algebra
\[
\mathcal C=\bigotimes_{i\geq1}2^{\{0,1\}}.
\]
\end{definition}

This is a product, not a union: a countable union of copies of $\{0,1\}$ is still $\{0,1\}$, whereas $\{0,1\}^{\Naturals}$ contains complete binary sequences.

\begin{proposition}[Well-defined Joint Law]\label{prop:joint-law}
$\mathbf Y_{X,T}$ is $\mathfrak A^{\otimes\Naturals}/\mathcal C$-measurable.  Its coordinates are i.i.d.\ Bernoulli with parameter
\[
\mu_X(T)=\E_\rho[Y(Z;X,T)].
\]
\end{proposition}

\begin{proof}
The product sigma-algebra $\mathcal C$ is generated by cylinder sets depending on finitely many coordinates.  The inverse image of each such cylinder is a finite intersection of events of the form $\{Y_i=a_i\}$, measurable by Proposition~\ref{prop:verdict-measurable}.  Hence the joint map is measurable.  The $Z_i$ are independent with common law $\rho$, and applying the same measurable binary function to each coordinate preserves independence and identical distribution.  A binary variable is Bernoulli with parameter equal to its expectation.
\end{proof}

\begin{remark}[Same Model Does Not Imply Dependence]\label{rem:same-model}
If every episode uses one fixed model configuration $\theta_0$, replace $\nu$ by the point mass $\delta_{\theta_0}$.  Independent seeds and private evidence still make the $Z_i$, hence the $Y_i$, independent and identically distributed.  The common model may create common \emph{bias} through the value of $\mu_X(T)$; it does not create statistical dependence between deterministic functions of independent inputs.  A shared random external state would be a different model and can create dependence.
\end{remark}

\begin{definition}[Population Decision and Clarity]\label{def:population-decision}
Fix an inclusive threshold $\tau\in(0,1)$.  The superpopulation decision and pointwise clarity are
\[
D_{X,\infty}(T)=\ind[\mu_X(T)\geq\tau],
\qquad
\gamma_X(T)=|\mu_X(T)-\tau|.
\]
At the boundary $\gamma_X(T)=0$, the target is defined by the tie policy, but finite i.i.d.\ panel decisions need not stabilize on one side.
\end{definition}

The countable product construction defines the infinite-dimensional object.  A statistic based only on $\sum_{i=1}^K Y_i$ is nevertheless scalar.  Supplement~\ref{supp:functional-limits} retains the complete partial-sum path and thereby obtains a genuinely functional limit.

\begin{remark}[Finite Censuses Are a Separate Regime]\label{rem:finite-separate}
The product model above describes i.i.d.\ evaluator episodes from a declared measure.  A concrete catalogue of distinct evaluators sampled without replacement is not i.i.d.\ conditional on its frozen verdicts.  We now construct that design directly, without importing independence from the superpopulation regime.
\end{remark}

\section{Exact Finite-Population Results}\label{sec:finite-population}

The finite model conditions on what a declared evaluator census actually returned.  Once these verdicts are frozen, the only randomness is the sampling design.

\subsection{Census and Nested Panels}

\begin{definition}[Finite Evaluator Census]\label{def:finite-census}
For fixed $(X,T)$, a census of size $M$ is a finite sequence of eligible evaluator episodes together with their frozen global verdicts
\[
\mathbf y_M(X,T)=(y_{M,1},\ldots,y_{M,M})\in\{0,1\}^M.
\]
Define
\[
C_M(T)=\sum_{i=1}^M y_{M,i},
\qquad
\mu_M(T)=\frac{C_M(T)}{M}.
\]
The dependence on $X$ is suppressed when one problem is fixed.
\end{definition}

No probability model is imposed on $\mathbf y_M$.  Evaluators may use different models, prompts, or evidence; their frozen verdicts may exhibit any pattern.

Let $\Pi_M$ be a uniform random permutation of $\{1,\ldots,M\}$.  For $1\leq K\leq M$, define the nested panel
\[
S_{M,K}=\{\Pi_M(1),\ldots,\Pi_M(K)\}
\]
and its number of positive global verdicts
\[
H_{M,K}(T)=\sum_{i=1}^K y_{M,\Pi_M(i)}.
\]
Each $S_{M,K}$ is a uniform $K$-element subset, and
\[
S_{M,1}\subset S_{M,2}\subset\cdots\subset S_{M,M}.
\]
The nesting is not needed for a single endpoint probability.  It is needed to model escalation by adding evaluators and to define the functional path in Supplement~\ref{supp:functional-limits}.

\begin{definition}[Inclusive Threshold Decisions]\label{def:finite-decisions}
For threshold $\tau\in(0,1)$, let
\[
q_K=\lceil\tau K\rceil.
\]
The panel and census decisions are
\[
D_{M,K}(T)=\ind[H_{M,K}(T)\geq q_K],
\qquad
D_M(T)=\ind[C_M(T)\geq q_M]
=\ind[\mu_M(T)\geq\tau].
\]
The census clarity is
\[
\gamma_M(T)=|\mu_M(T)-\tau|.
\]
\end{definition}

The inclusive convention means that an exact tie at the threshold accepts.  With $\tau=1/2$ and even $K$, this differs from strict majority.  The convention is part of the mechanism and must not be altered after observing data.

\subsection{Exact Panel Law}

\begin{theorem}[Exact Hypergeometric Law]\label{thm:hypergeometric}
Conditional on the frozen census $\mathbf y_M(X,T)$,
\[
H_{M,K}(T)\sim\operatorname{Hypergeometric}(M,C_M(T),K),
\]
that is,
\[
\Prob(H_{M,K}=h\mid\mathbf y_M)
=
\frac{\binom{C_M}{h}\binom{M-C_M}{K-h}}{\binom{M}{K}}
\]
for every feasible integer $h$.
\end{theorem}

\begin{proof}
The first $K$ entries of a uniform permutation form a uniform $K$-subset.  There are $\binom{M}{K}$ such subsets.  A subset with exactly $h$ positive entries chooses $h$ of the $C_M$ positives and $K-h$ of the $M-C_M$ negatives, giving the numerator.
\end{proof}

\begin{definition}[Panel--Census Error]\label{def:panel-census-error}
The exact conditional error is
\[
e_{M,K}(T)
=
\Prob\bigl(D_{M,K}(T)\neq D_M(T)\mid\mathbf y_M(X,T)\bigr).
\]
\end{definition}

By Theorem~\ref{thm:hypergeometric},
\[
e_{M,K}(T)=
\begin{cases}
\displaystyle
\sum_{h=0}^{q_K-1}
\frac{\binom{C_M}{h}\binom{M-C_M}{K-h}}{\binom{M}{K}},
& D_M(T)=1,\\[2.0ex]
\displaystyle
\sum_{h=q_K}^{K}
\frac{\binom{C_M}{h}\binom{M-C_M}{K-h}}{\binom{M}{K}},
& D_M(T)=0,
\end{cases}
\]
where infeasible terms are zero.  Thus no normal approximation or Monte Carlo simulation is required to certify a uniform finite panel once $C_M$ is known.

\begin{proposition}[Mean and Finite-Population Correction]\label{prop:finite-moments}
Let $\widehat\mu_{M,K}=H_{M,K}/K$.  Conditional on the census,
\[
\E[\widehat\mu_{M,K}\mid\mathbf y_M]=\mu_M
\]
and, for $M>1$,
\[
\Var(\widehat\mu_{M,K}\mid\mathbf y_M)
=
\frac{\mu_M(1-\mu_M)}{K}\,
\frac{M-K}{M-1}.
\]
\end{proposition}

\begin{proof}
Write $I_i=\ind[i\in S_{M,K}]$.  Then $\E I_i=K/M$ and
\[
\Cov(I_i,I_j)=-\frac{K(M-K)}{M^2(M-1)},\qquad i\neq j.
\]
Substitution in $H_{M,K}=\sum_i I_i y_{M,i}$ gives the stated hypergeometric moments.
\end{proof}

For $M=1500$, the standard-deviation correction relative to independent sampling is
\[
\sqrt{\frac{1500-K}{1499}}.
\]
It is approximately $0.997$, $0.984$, and $0.966$ for $K=10,50,100$, respectively, and equals zero at the full census.  A binomial calculation may therefore be numerically close for a small sampling fraction while still answering a different conditional question.

\subsection{Finite-Sample Certificates}

Serfling's inequality gives a distribution-free finite-population bound~\cite{serfling1974probability}.  We state the binary specialization.

\begin{theorem}[Panel--Census Concentration]\label{thm:serfling-panel}
If $\gamma_M(T)>0$, then
\[
e_{M,K}(T)
\leq
\exp\!\left(
-\frac{2K\gamma_M(T)^2}{1-(K-1)/M}
\right).
\]
\end{theorem}

\begin{proof}
If $D_M=1$ and $\mu_M>\tau$, an error implies
$\widehat\mu_{M,K}-\mu_M\leq-\gamma_M$.
If $D_M=0$, an error implies
$\widehat\mu_{M,K}-\mu_M\geq\gamma_M$.
Indeed, because $H_{M,K}$ is integer,
$H_{M,K}<q_K=\lceil\tau K\rceil$ implies
$\widehat\mu_{M,K}<\tau$, whereas $H_{M,K}\geq q_K$ implies
$\widehat\mu_{M,K}\geq\tau$.
Apply the corresponding one-sided Serfling inequality to a population in $[0,1]$.  The case $\mu_M=\tau$ is excluded because then $\gamma_M=0$.
\end{proof}

Dropping the finite-population improvement gives the conservative sufficient condition
\[
K\geq\frac{\log(1/\delta)}{2\gamma_M(T)^2}
\quad\Longrightarrow\quad
e_{M,K}(T)\leq\delta,
\]
provided $K\leq M$.  Exact hypergeometric inversion is preferable whenever $M$ and $C_M$ are available.

\begin{definition}[Minimum Certified Panel Size]\label{def:kmin}
For an error target $\delta\in(0,1)$,
\[
K_{\min}(T;M,\delta)
=
\min\{1\leq K\leq M:e_{M,K}(T)\leq\delta\}.
\]
\end{definition}

Because integer quotas and tie conventions can create parity effects, $e_{M,K}$ need not decrease at every consecutive $K$.  Definition~\ref{def:kmin} uses the exact sequence rather than assuming monotonicity.  If a deployment requires that every larger nested panel also meet the target, it should instead use
\[
K_{\mathrm{stable}}(T;M,\delta)
=
\min\{K:e_{M,k}(T)\leq\delta\text{ for every }k=K,\ldots,M\}.
\]

Figure~\ref{fig:finite-error} in the main text compares these exact tails with
their small-sampling-fraction binomial approximations.

\paragraph{Scope of the finite result.}
The exact law certifies agreement with the specified census.  It does not by itself certify that the census represents a larger population, that the population is externally correct, or that a new evaluation campaign under changed conditions will reproduce the same frozen vector.  These are separate layers developed in Supplement~\ref{supp:functional-limits} and Appendix~\ref{sec:population-limits}.

\section{Task-Specific Population Limits}\label{sec:population-limits}

The phrase ``all possible interpreters'' does not define a probability distribution.  An ADS must declare how a finite evaluator design grows and, if an infinite idealization is used, which limit that design is intended to approximate.  The declaration may depend on the problem $X$---for example, legal, mathematical, and medical problems may require different evaluator populations---but it must not be chosen retrospectively to favor a submitted resolution.

In this section, $\nu_X$ denotes a task-specific law on the complete episode space $\calZ$.  It specializes the generic episode law $\rho$ from Section~\ref{sec:model}; equivalently, it can be built from task-specific configuration weights and the declared private-randomness law.

\subsection{Growing Finite Designs}

Fix a problem $X$ and let $M_1<M_2<\cdots$ tend to infinity.  For each $n$, let
\[
\census_{X,n}=(z_{n,1},\ldots,z_{n,M_n})
\]
be a declared finite collection of complete evaluator episodes.  Repeated configurations are allowed when their assigned weights require them.  Its empirical design measure is
\[
\nu_{X,n}=\frac1{M_n}\sum_{i=1}^{M_n}\delta_{z_{n,i}}.
\]
For resolution $T$, define
\[
\mu_{X,n}(T)
=
\int_{\calZ}Y(z;X,T)\,d\nu_{X,n}(z)
=
\frac1{M_n}\sum_{i=1}^{M_n}Y(z_{n,i};X,T).
\]
The associated finite-census decision is
\[
D_{X,M_n}(T)=\ind[\mu_{X,n}(T)\geq\tau].
\]
A nested expansion study may require each collection to be a prefix of one precommitted evaluator design.  The convergence results below only require the displayed sequence of empirical measures; nesting is an additional experimental constraint, not a hidden mathematical assumption.

\begin{definition}[Admissible Limiting Population]\label{def:admissible-limit}
For a declared class of resolutions $\calT_X$, the sequence $(\nu_{X,n})$ is ADS-admissible with limit $\nu_X$ if $\nu_X$ is a probability measure on $(\calZ,\mathfrak A)$ and
\[
\mu_{X,n}(T)
\longrightarrow
\mu_X(T)
:=
\int_{\calZ}Y(z;X,T)\,d\nu_X(z)
\]
for every $T\in\calT_X$.
\end{definition}

Weak convergence $\nu_{X,n}\Rightarrow\nu_X$ alone is not always sufficient because the binary map $z\mapsto Y(z;X,T)$ may be discontinuous.  Definition~\ref{def:admissible-limit} therefore states the integral convergence actually needed by the decision system.  Weak convergence plus $\nu_X$-almost-everywhere continuity of the vote map is one sufficient route.

\begin{definition}[Limiting Decision]\label{def:limiting-decision}
For an admissible population and inclusive threshold $\tau$,
\[
D_{X,\infty}(T)=\ind[\mu_X(T)\geq\tau],
\qquad
\gamma_X(T)=|\mu_X(T)-\tau|.
\]
\end{definition}

\begin{theorem}[Deterministic Census Consistency]\label{thm:deterministic-consistency}
If $(\nu_{X,n})$ is ADS-admissible and $\gamma_X(T)>0$, then there exists $n_0(T)$ such that
\[
D_{X,M_n}(T)=D_{X,\infty}(T)
\qquad\text{for every }n\geq n_0(T).
\]
\end{theorem}

\begin{proof}
Let $\gamma=|\mu_X(T)-\tau|>0$.  By convergence, eventually
$|\mu_{X,n}(T)-\mu_X(T)|<\gamma$.  Hence $\mu_{X,n}(T)$ and $\mu_X(T)$ lie on the same side of $\tau$; with a strict inequality one may use $\gamma/2$ to avoid the boundary explicitly.
\end{proof}

Convergence alone supplies no numerical value of $n_0$.  A finite census can be certified against the limit only after a rate or error envelope has been supplied.

\begin{proposition}[Rate-Based Limit Certificate]\label{prop:rate-certificate}
Suppose a justified bound
\[
|\mu_{X,n}(T)-\mu_X(T)|\leq b_{X,n}(T)
\]
is available.  If the observable census margin satisfies
\[
|\mu_{X,n}(T)-\tau|>b_{X,n}(T),
\]
then $D_{X,M_n}(T)=D_{X,\infty}(T)$ and
\[
\gamma_X(T)
\geq
|\mu_{X,n}(T)-\tau|-b_{X,n}(T)>0.
\]
\end{proposition}

\begin{proof}
The reverse triangle inequality gives
\[
|\mu_X(T)-\tau|
\geq
|\mu_{X,n}(T)-\tau|-|\mu_{X,n}(T)-\mu_X(T)|>0.
\]
Moreover, the interval
\[
[\mu_{X,n}(T)-b_{X,n}(T),\mu_{X,n}(T)+b_{X,n}(T)]
\]
lies entirely on one side of $\tau$ and contains $\mu_X(T)$.  The finite and limiting means therefore induce the same threshold decision.
\end{proof}

For a family of resolutions, a uniform envelope
\[
\sup_{T\in\calT_X}|\mu_{X,n}(T)-\mu_X(T)|\leq b_{X,n}
\]
certifies every observed resolution whose census margin exceeds $b_{X,n}$.  Obtaining such a uniform envelope requires complexity assumptions on the resolution class; pointwise convergence alone does not provide it.

\subsection{Probabilistic Superpopulations}

A different construction draws $Z_1,\ldots,Z_M$ i.i.d.\ from a declared task-specific measure $\nu_X$.  Then, for fixed $T$,
\[
Y(Z_i;X,T)\overset{\mathrm{i.i.d.}}{\sim}
\operatorname{Bernoulli}(\mu_X(T)),
\qquad
C_M(T)\sim\operatorname{Binomial}(M,\mu_X(T)).
\]
The census estimator $\widehat\mu_M=C_M/M$ obeys
\[
\Prob\bigl(|\widehat\mu_M-\mu_X(T)|\geq\varepsilon\bigr)
\leq2e^{-2M\varepsilon^2}
\]
by Hoeffding's inequality~\cite{hoeffding1963probability}, and exact confidence intervals can be obtained by inverting the binomial distribution.

\begin{proposition}[Unconditional Subsample Law]\label{prop:unconditional-subsample}
Generate an i.i.d.\ census $Z_1,\ldots,Z_M\sim\nu_X$ and, independently, select a uniform $K$-subset of its distinct indices.  Unconditionally, the selected evaluator episodes are i.i.d.\ with law $\nu_X$.  Consequently their positive count is
\[
\operatorname{Binomial}(K,\mu_X(T)).
\]
Conditional on the realized census verdicts, the same count is hypergeometric as in Theorem~\ref{thm:hypergeometric}.
\end{proposition}

\begin{proof}
Condition on any ordered tuple of $K$ distinct indices.  The corresponding coordinates of an i.i.d.\ sample are independent with common law $\nu_X$, and this conditional joint law does not depend on the chosen tuple.  Mixing over the independently selected indices leaves the product law $\nu_X^{\otimes K}$.  Conditioning instead on all realized binary coordinates fixes the total count and yields uniform sampling without replacement.
\end{proof}

This proposition resolves an apparent contradiction: the finite-population correction is a conditional statement about a realized census, while the binomial law is an unconditional statement about a randomly generated census.

At $M=1500$, the largest possible asymptotic standard error of $\widehat\mu_M$ is
\[
\frac{1}{2\sqrt{1500}}\approx0.0129.
\]
This is an approximation to estimator variability, not a universal decision guarantee.  If $\mu_X(T)$ lies arbitrarily close to $\tau$, even a large census can fall on the other side with substantial probability.  Numerical claims must therefore report the estimated clarity as well as $M$.

\subsection{Two-Level Error Accounting}

When a finite census is used as a proxy for a limiting population, two discrepancies must be separated:
\begin{enumerate}[nosep]
    \item panel versus realized census;
    \item realized census versus limiting population.
\end{enumerate}
Conditional on a realized census $\mathbf y_M$, the deterministic triangle inequality for disagreement indicators gives
\[
\Prob(D_{M,K}\neq D_{X,\infty}\mid\mathbf y_M)
\leq
e_{M,K}(T)
+
\ind[D_M\neq D_{X,\infty}].
\]
The first term is the exact conditional hypergeometric error.  The second requires either a deterministic rate such as Proposition~\ref{prop:rate-certificate} or a probabilistic model for census generation.  If the census is random, integrating the displayed bound yields the corresponding unconditional decomposition.  Neither form can be obtained from the phrase ``large population'' alone.

\section{Population and Empirical Resolvability of a Problem Family}\label{sec:resolvability}

One problem can produce many different concrete resolutions.  The acceptance mean and panel error must be calculated separately for each resolution before any aggregation across generators.  We first define the ideal quantity under the declared task-specific population and then show how growing finite censuses approximate it.

\subsection{Generator and Evaluator Populations}

Let $(\mathcal U,\mathfrak U,\pi)$ be a declared population of generator configurations and let the measurable map, with $\calT_X$ carrying its discrete sigma-algebra,
\[
\operatorname{Solve}:\mathcal U\times\{X\}\longrightarrow\calT_X
\]
produce a frozen resolution $T_u=\operatorname{Solve}(u,X)$.  A generator configuration includes model, prompt, decoding parameters, seed, tools, and environment policy.  The evaluator population $\nu_X$ is conceptually distinct from $\pi$, even when the same catalogue of LLM configurations is used in both roles.

For finite generator and evaluator censuses $\mathcal U_A=\{u_1,\ldots,u_A\}$ and $\mathcal B_M=\{z_1,\ldots,z_M\}$, write $T_a=\operatorname{Solve}(u_a,X)$ and define the global evaluation matrix
\[
\mathbf Y=(Y_{b,a})\in\{0,1\}^{M\times A},
\qquad
Y_{b,a}=Y(z_b;X,T_a).
\]
Each column is one concrete resolution; each row is one evaluator episode.  Once constructed, $\mathbf Y$ is deterministic.  If generator and evaluator catalogues overlap, self-evaluation must be excluded or reported separately according to a predeclared policy.

\subsection{Ideal Pointwise and Family-Level Definitions}

Under the probabilistic population $\nu_X$, a fresh panel of size $K$ consists of i.i.d.\ evaluator episodes.  For fixed $T$, let
\[
B_{X,K}(T)\sim\operatorname{Binomial}(K,\mu_X(T)),
\qquad
D_{X,K}(T)=\ind[B_{X,K}(T)\geq q_K].
\]

\begin{definition}[Ideal Pointwise Panel Error]\label{def:ideal-panel-error}
The panel error relative to the limiting population decision is
\[
e_{X,K}(T)
=
\Prob\bigl(D_{X,K}(T)\neq D_{X,\infty}(T)\bigr).
\]
Equivalently,
\[
e_{X,K}(T)=
\begin{cases}
\Prob(B_{X,K}(T)<q_K),&D_{X,\infty}(T)=1,\\
\Prob(B_{X,K}(T)\geq q_K),&D_{X,\infty}(T)=0.
\end{cases}
\]
\end{definition}

This is an exact binomial tail, not a CLT approximation.  Sampling without replacement from a countably infinite set has no uniform law.  The expression above instead has either of two precise interpretations: direct i.i.d.\ sampling from $\nu_X$, or the fixed-$K$ limit of uniform without-replacement panels from an admissible sequence of finite censuses, proved below.

\begin{definition}[Population Resolvability Coverage]\label{def:population-coverage}
The ideal coverage of a panel size $K$ at error tolerance $\delta$ is
\[
R_{X;K,\delta}
=
\int_{\mathcal U}
\ind[e_{X,K}(\operatorname{Solve}(u,X))\leq\delta]\,d\pi(u).
\]
The problem $X$ is $(K,\delta,\beta)$-resolvable relative to $(\pi,\nu_X,G,\tau)$ if
\[
R_{X;K,\delta}\geq1-\beta.
\]
\end{definition}

Thus a generator draw produces a resolution for which an independent $K$-evaluator panel reproduces the limiting population decision with probability at least $1-\delta$, except on a generator mass of at most $\beta$.  This definition does not average the verdicts of different resolutions.

\subsection{Finite Matrix and Empirical Coverage}

For each column $a$, let
\[
C_a=\sum_{b=1}^M Y_{b,a},
\qquad
\mu_M(T_a)=\frac{C_a}{M},
\qquad
D_M(T_a)=\ind[\mu_M(T_a)\geq\tau].
\]
Let $e_{M,K}(T_a)$ be the exact hypergeometric panel--census error from Definition~\ref{def:panel-census-error}.

\begin{definition}[Empirical Pointwise Criterion]\label{def:pointwise-resolvable}
A resolution $T$ is empirically $(K,\delta)$-resolvable relative to the declared evaluator census if
\[
e_{M,K}(T)\leq\delta.
\]
\end{definition}

\begin{definition}[Empirical Resolvability Coverage]\label{def:coverage}
For equally weighted generated resolutions,
\[
R_{K,\delta}^{(A,M)}
=
\frac1A\sum_{a=1}^A
\ind[e_{M,K}(T_a)\leq\delta].
\]
For declared generator weights $r_a\geq0$, $\sum_a r_a=1$, use
\[
R_{K,\delta}^{(r,M)}
=
\sum_{a=1}^A r_a
\ind[e_{M,K}(T_a)\leq\delta].
\]
\end{definition}

\begin{definition}[Empirically $(K,\delta,\beta)$-Resolvable Problem]\label{def:problem-resolvable}
The finite experimental problem is $(K,\delta,\beta)$-resolvable if
\[
R_{K,\delta}^{(A,M)}\geq1-\beta.
\]
Thus at least a fraction $1-\beta$ of generated resolutions can be classified by a panel of size $K$ with panel--census disagreement at most $\delta$.
\end{definition}

The definition is explicitly relative to the generator population, evaluator population, global rule, threshold, tie policy, and sampling design.  These objects are part of the ADS, not nuisance details.

Pooling all entries of $\mathbf Y$ before thresholding destroys essential information.  For example, if half the resolutions have population acceptance $0.7$ and half have $0.3$, their pooled mean is $0.5$ although every individual resolution has clarity $0.2$ at threshold $0.5$.

\subsection{An End-to-End Finite-Sample Certificate}\label{subsec:operational-certificate}

The empirical criterion above is exact relative to a fully observed evaluator census.  To infer ideal coverage $R_{X;K,\delta}$ from a finite random matrix, the uncertainty in its rows and columns must also be accounted for.  We now give a simultaneous certificate that does so without a CLT.

For $p\in[0,1]$, define the two directed binomial tails
\[
e_K^{\mathrm{acc}}(p)
=
\Prob(\operatorname{Binomial}(K,p)<q_K),
\qquad
e_K^{\mathrm{rej}}(p)
=
\Prob(\operatorname{Binomial}(K,p)\geq q_K).
\]
The first is non-increasing in $p$ and the second is non-decreasing.  For an interval $I=[L,U]\subseteq[0,1]$, set
\[
\operatorname{Cert}_{K,\delta}(I)
=
\ind\!\left[
\begin{array}{l}
L\geq\tau\ \text{and}\ e_K^{\mathrm{acc}}(L)\leq\delta,\\[-0.1em]
\hfill\text{or}\\[-0.1em]
U<\tau\ \text{and}\ e_K^{\mathrm{rej}}(U)\leq\delta.
\end{array}
\right]
\]
An interval that crosses the decision threshold is deliberately left uncertified.

\begin{lemma}[Sound Interval Certificate]\label{lem:interval-certificate}
For every resolution $T$ and interval $I$ satisfying
$\mu_X(T)\in I$,
\[
\operatorname{Cert}_{K,\delta}(I)
\leq
\ind[e_{X,K}(T)\leq\delta].
\]
\end{lemma}

\begin{proof}
Couple Bernoulli variables at parameters $p\leq p'$ by using common uniforms: $\ind[\zeta_i\leq p]\leq\ind[\zeta_i\leq p']$.  Their sums are ordered almost surely, so $e_K^{\mathrm{acc}}$ is non-increasing and $e_K^{\mathrm{rej}}$ is non-decreasing.  If the first branch of the certificate holds, every $p\in I$ satisfies $p\geq L\geq\tau$, hence the limiting target accepts and
\[
e_{X,K}(T)=e_K^{\mathrm{acc}}(p)
\leq e_K^{\mathrm{acc}}(L)\leq\delta.
\]
The second branch is symmetric: every $p\in I$ satisfies $p\leq U<\tau$, the target rejects, and
\[
e_{X,K}(T)=e_K^{\mathrm{rej}}(p)
\leq e_K^{\mathrm{rej}}(U)\leq\delta.
\]
\end{proof}

Fix a non-empty finite set $\mathcal K\subset\Naturals$ of candidate deployment-panel sizes and a pointwise tolerance $\delta$.  The set $\mathcal K$ and $\delta$ are fixed before inspecting the matrix.  Let
\[
u_1,\ldots,u_A\overset{\mathrm{i.i.d.}}{\sim}\pi,
\qquad
Z_1,\ldots,Z_M\overset{\mathrm{i.i.d.}}{\sim}\nu_X,
\]
with the two samples independent, and construct
\[
T_a=\operatorname{Solve}(u_a,X),
\qquad
C_a=\sum_{b=1}^M Y(Z_b;X,T_a).
\]
The same evaluator rows may be used in every column.  Each $Z_b$ denotes the
complete random element for row $b$ and can induce arbitrary dependence among
its cells; the requirement is that the rows are i.i.d.\ across $b$.  Given
$T_a$, each marginal count is
$\operatorname{Binomial}(M,\mu_X(T_a))$; no cross-column independence is
asserted or needed.  Random state persistent across rows is not absorbed into
private cell randomness: it must be conditioned on or incorporated into a
different sampling model.

Choose error budgets $\eta_E,\eta_G\in(0,1)$ with $\eta_E+\eta_G<1$.  For $s\in\{0,\ldots,A\}$ and $\alpha\in(0,1)$, write
\[
\ell_A(s;\alpha)
=
\begin{cases}
0, & s=0,\\
\operatorname{Beta}^{-1}(\alpha;s,A-s+1), & s\geq1,
\end{cases}
\]
for the one-sided Clopper--Pearson lower limit after $s$ successes in $A$ binomial trials~\cite{clopper1934use}.  It is non-decreasing in $s$ and satisfies
\[
\sup_{q\in[0,1]}
\Prob_{S\sim\operatorname{Binomial}(A,q)}
\bigl(q<\ell_A(S;\alpha)\bigr)
\leq\alpha.
\]
For each column, let $I_a=I_{M,\eta_E/A}(C_a)$ be any binomial confidence interval satisfying
\[
\inf_{p\in[0,1]}
\Prob_{C\sim\operatorname{Binomial}(M,p)}
\bigl(p\in I_{M,\eta_E/A}(C)\bigr)
\geq1-\frac{\eta_E}{A}.
\]
Exact Clopper--Pearson intervals are one choice; a closed-form Hoeffding alternative is given below.  Define
\[
\widehat R_{K,\delta}^{\mathrm{cert}}
=
\frac1A\sum_{a=1}^A
\operatorname{Cert}_{K,\delta}(I_a)
,
\qquad
S_{K}^{\mathrm{cert}}
=
A\widehat R_{K,\delta}^{\mathrm{cert}},
\]
and
\[
\underline R_{K,\delta}^{\mathrm{cert}}
=
\ell_A\!\left(S_K^{\mathrm{cert}};
\frac{\eta_G}{|\mathcal K|}\right).
\]

\begin{theorem}[Familywise End-to-End Resolvability Certificate]\label{thm:operational-certificate}
Under the sampling design above,
\[
\Prob\!\left(
R_{X;K,\delta}\geq\underline R_{K,\delta}^{\mathrm{cert}}
\text{ for every }K\in\mathcal K
\right)
\geq1-\eta_E-\eta_G.
\]
Consequently, any data-dependent choice $\widehat K\in\mathcal K$ satisfying
\[
\underline R_{\widehat K,\delta}^{\mathrm{cert}}\geq1-\beta
\]
certifies that $X$ is $(\widehat K,\delta,\beta)$-resolvable with confidence at least $1-\eta_E-\eta_G$.
\end{theorem}

\begin{proof}
Condition on the complete generator sample.  Each interval has conditional miscoverage at most $\eta_E/A$, even though the column counts can be dependent through their shared evaluator rows.  A union bound therefore gives
\[
\Prob\bigl(\mu_X(T_a)\in I_a\text{ for every }a=1,\ldots,A\bigr)
\geq1-\eta_E.
\]
On this event, Lemma~\ref{lem:interval-certificate} gives, simultaneously for every $a$ and $K\in\mathcal K$,
\[
\operatorname{Cert}_{K,\delta}(I_a)
\leq
W_{a,K}
:=
\ind[e_{X,K}(T_a)\leq\delta].
\]

For fixed $K$, the variables $W_{1,K},\ldots,W_{A,K}$ are i.i.d.\ Bernoulli with mean $R_{X;K,\delta}$.  Put $S_K^W=\sum_a W_{a,K}$.  Exact one-sided binomial inversion gives
\[
\Prob\!\left(
R_{X;K,\delta}
<\ell_A\!\left(S_K^W;\frac{\eta_G}{|\mathcal K|}\right)
\right)
\leq\frac{\eta_G}{|\mathcal K|}.
\]
A second union bound makes this statement simultaneous over $\mathcal K$ with failure probability at most $\eta_G$.  On the evaluator event, $S_K^{\mathrm{cert}}\leq S_K^W$.  Monotonicity of $\ell_A$ therefore gives, on the intersection of the evaluator- and generator-level events,
\[
R_{X;K,\delta}
\geq
\ell_A\!\left(S_K^W;\frac{\eta_G}{|\mathcal K|}\right)
\geq
\ell_A\!\left(S_K^{\mathrm{cert}};\frac{\eta_G}{|\mathcal K|}\right)
\]
for every candidate $K$.  A final union bound gives the stated confidence, and simultaneity permits selection of $\widehat K$ from the same matrix.
\end{proof}

The theorem keeps four logically different quantities separate.  The tolerance $\delta$ controls fresh-panel disagreement for each certified resolution; $\beta$ controls the generator mass of uncertified resolutions; $\eta_E$ controls estimation of the evaluator means; and $\eta_G$ controls generalization from sampled generators.  The last two are confidence budgets for the calibration study, not runtime panel-error probabilities.

The familywise construction is strongest pointwise---every reported sampled column is sound on one common event---but its interval level $\eta_E/A$ becomes conservative when $A$ is large.  For population resolvability, it is enough to control the \emph{generator mass} of interval failures.  The next result makes that tradeoff explicit.

Fix an evaluator-slack budget $\xi_E\in(0,1)$.  For every observed column, now use the pointwise interval
\[
I_a^{\mathrm{mass}}
=
I_{M,\eta_E\xi_E}(C_a),
\]
and define
\[
\widehat R_{K,\delta}^{\mathrm{mass}}
=
\frac1A\sum_{a=1}^A
\operatorname{Cert}_{K,\delta}(I_a^{\mathrm{mass}}),
\qquad
S_K^{\mathrm{mass}}
=
A\widehat R_{K,\delta}^{\mathrm{mass}},
\]
and
\[
\underline R_{K,\delta}^{\mathrm{mass}}
=
\max\{0,
\ell_A\!\left(S_K^{\mathrm{mass}};
\frac{\eta_G}{|\mathcal K|}\right)
-\xi_E\}.
\]

The integrals used below are well defined.  Because $\calT_X$ is countable and carries its discrete sigma-algebra, the fixed-input measurability in Definition~\ref{def:interpreter-family} implies joint measurability of $(z,T)\mapsto Y(z;X,T)$: inverse images are countable unions of measurable fixed-$T$ sections.  Hence $u\mapsto\mu_X(\operatorname{Solve}(u,X))$ is measurable by composition and parameter integration.  The count, interval, and certificate maps are measurable finite-range functions, so the random masses in the proof are measurable and Tonelli's theorem applies.

\begin{proof}[Proof of Theorem~\ref{thm:mass-operational-certificate}]
Write $\mathbf Z=(Z_1,\ldots,Z_M)$.  For every generator configuration $u\in\mathcal U$, including configurations not sampled into the observed matrix, let
\[
I(u;\mathbf Z)
=
I_{M,\eta_E\xi_E}\!\left(
\sum_{b=1}^M
Y(Z_b;X,\operatorname{Solve}(u,X))
\right).
\]
Define the random generator mass whose evaluator interval misses its mean by
\[
F(\mathbf Z)
=
\int_{\mathcal U}
\ind\!\left[
\mu_X(\operatorname{Solve}(u,X))\notin I(u;\mathbf Z)
\right]d\pi(u).
\]
For each fixed $u$, binomial interval validity gives miss probability at most $\eta_E\xi_E$.  Tonelli's theorem therefore yields
\[
\E_{\mathbf Z}[F(\mathbf Z)]\leq\eta_E\xi_E,
\]
and Markov's inequality gives
\[
\Prob_{\mathbf Z}(F(\mathbf Z)>\xi_E)\leq\eta_E.
\]

For fixed $\mathbf Z$ and $K$, let
\[
Q_K(\mathbf Z)
=
\int_{\mathcal U}
\operatorname{Cert}_{K,\delta}(I(u;\mathbf Z))\,d\pi(u).
\]
The observed certificate indicators are conditionally i.i.d.\ Bernoulli with mean $Q_K(\mathbf Z)$ because the generator sample is i.i.d.\ and independent of $\mathbf Z$.  Conditional exact binomial inversion followed by a union bound over $\mathcal K$ gives, with probability at least $1-\eta_G$,
\[
Q_K(\mathbf Z)
\geq
\ell_A\!\left(S_K^{\mathrm{mass}};
\frac{\eta_G}{|\mathcal K|}\right)
\qquad\text{for every }K\in\mathcal K.
\]
For every $u$ whose interval contains its mean, Lemma~\ref{lem:interval-certificate} makes a positive certificate sound.  Hence, for every $K$,
\[
Q_K(\mathbf Z)
\leq
R_{X;K,\delta}+F(\mathbf Z).
\]
On the intersection of $F(\mathbf Z)\leq\xi_E$ and the conditional generator event,
\[
R_{X;K,\delta}
\geq
\ell_A\!\left(S_K^{\mathrm{mass}};
\frac{\eta_G}{|\mathcal K|}\right)-\xi_E
\]
simultaneously over $\mathcal K$.  The union bound over the two failure budgets and clipping at zero complete the proof.
\end{proof}

The mass-controlled theorem replaces the familywise requirement ``no sampled interval fails'' by ``the generator mass of interval failures is at most $\xi_E$.''  It remains valid under arbitrary dependence across columns induced by shared evaluator rows.  Its pointwise interval level $\eta_E\xi_E$ is independent of $A$, while the price $\xi_E$ is visible in the final coverage bound and can be allocated as part of the unresolved-mass budget $\beta$.

For a closed-form alternative at the generator layer, define
\[
\varepsilon_G
=
\sqrt{\frac{\log(|\mathcal K|/\eta_G)}{2A}}.
\]
Hoeffding's inequality makes
\[
\max\{0,\widehat R_{K,\delta}^{\mathrm{cert}}-\varepsilon_G\}
\quad\text{and}\quad
\max\{0,\widehat R_{K,\delta}^{\mathrm{mass}}-\xi_E-\varepsilon_G\}
\]
valid simultaneous lower bounds for the familywise and mass-controlled constructions, respectively.  The exact inversion above is the reported primary bound; the Hoeffding form is retained as a transparent conservative benchmark.

Both theorems extend to a predeclared finite grid $\mathcal J$ of pairs $(K,\delta)$ by replacing $|\mathcal K|$ with $|\mathcal J|$ in $\varepsilon_G$ and taking the generator-level union bound over $\mathcal J$.  This permits simultaneous reporting or data-dependent selection of both quantities.  Choosing a new grid or a new value of $\xi_E$ after inspecting the matrix requires a fresh multiplicity adjustment.

\begin{corollary}[Fresh-Resolution Deployment Error]\label{cor:deployment-error}
On the confidence event in either Theorem~\ref{thm:operational-certificate} or Theorem~\ref{thm:mass-operational-certificate}, if the corresponding lower bound is at least $1-\beta$, then an independent generator draw $u^\star\sim\pi$ followed by an independent $K$-evaluator panel satisfies
\[
\Prob\bigl(D_{X,K}(T_{u^\star})\neq D_{X,\infty}(T_{u^\star})\bigr)
\leq
(1-\beta)\delta+\beta
=
\delta+(1-\delta)\beta.
\]
\end{corollary}

\begin{proof}
Integrate the pointwise error $e_{X,K}(T_{u^\star})$.  It is at most $\delta$ on generator mass at least $1-\beta$ and is always at most one on the remaining mass.
\end{proof}

For a completely explicit but typically wider evaluator interval, let
\[
\widehat p_a=\frac{C_a}{M},
\qquad
b_E=\sqrt{\frac{\log(2A/\eta_E)}{2M}},
\qquad
I_a^{\mathrm H}
=
[\widehat p_a-b_E,\widehat p_a+b_E]\cap[0,1].
\]
Hoeffding's inequality and a union bound give simultaneous coverage of all $A$ evaluator means with probability at least $1-\eta_E$.  Thus $I_a^{\mathrm H}$ can be substituted directly into Theorem~\ref{thm:operational-certificate}.  Exact binomial inversion avoids some of this conservatism.

For Theorem~\ref{thm:mass-operational-certificate}, the corresponding pointwise Hoeffding radius is
\[
b_E^{\mathrm{mass}}
=
\sqrt{\frac{\log(2/(\eta_E\xi_E))}{2M}}.
\]
Unlike the familywise radius $b_E$, it does not grow with the number of generator columns.

The Hoeffding generator correction also exposes a simple sufficient experimental scale.  Even if every sampled resolution certifies, its lower bound can reach $1-\beta$ only if $\varepsilon_G\leq\beta$, for which the sufficient sizing rule is
\[
A\geq\frac{\log(|\mathcal K|/\eta_G)}{2\beta^2}.
\]
Under the mass-controlled certificate, choose $\xi_E<\beta$ and replace this by
\[
A\geq
\frac{\log(|\mathcal K|/\eta_G)}{2(\beta-\xi_E)^2}.
\]
For example, with $|\mathcal K|=8$ and $\eta_G=0.025$, this lower-limit calculation gives $A\geq513$ for $(\beta,\xi_E)=(0.10,0.025)$ and $A\geq1803$ for $(0.05,0.01)$.  These values only prevent a vacuous bound when every column certifies; a population with coverage close to the target requires additional power.
Beyond the explicit charges $\xi_E+\varepsilon_G$, power is reduced by the generator mass whose acceptance means lie within roughly one evaluator-interval radius of the pointwise certification boundary; such near-boundary columns certify only intermittently at finite $M$.
Failure to pass the certificate is inconclusive: it may reflect genuine ambiguity, too few evaluator rows, too few generator columns, or conservative simultaneous intervals.

\begin{corollary}[Finite Generator Catalogue]\label{cor:finite-generator-catalogue}
Suppose instead that the declared generator population is the fully enumerated finite law
\[
\pi=\sum_{a=1}^A r_a\delta_{u_a},
\qquad r_a\geq0,
\qquad \sum_{a=1}^A r_a=1,
\]
and every column is evaluated.  Construct the familywise intervals
$I_a=I_{M,\eta_E/A}(C_a)$.  Then, with probability at least $1-\eta_E$,
simultaneously for every $K\in\mathcal K$,
\[
R_{X;K,\delta}
\geq
\sum_{a=1}^A r_a\operatorname{Cert}_{K,\delta}(I_a).
\]
No generator-level Hoeffding correction is needed.
\end{corollary}

\begin{proof}
On simultaneous coverage of the evaluator intervals, Lemma~\ref{lem:interval-certificate} holds for every enumerated support point.  Multiply its inequalities by $r_a$ and sum.
\end{proof}

A mass-controlled version for the same finite generator law uses intervals at level $\eta_E\xi_E$ and subtracts $\xi_E$ from the displayed weighted sum.  Its proof is the same Tonelli--Markov argument as Theorem~\ref{thm:mass-operational-certificate}, with no generator-sampling step.

\begin{remark}[Workload-Level Extension]\label{rem:workload-extension}
The fixed problem $X$ is not essential to either matrix theorem.  Let a population unit $u$ encode a problem, a generator configuration, and its frozen resolution, and draw $u_1,\ldots,u_A$ i.i.d.\ from a declared workload--generator law $\Pi$.  For each $u$, let $p(u)$ be the acceptance mean under its declared task-specific evaluator law and let $e_K(u)$ denote the resulting exact panel-error tail under the corresponding $(G,\tau)$.  Define the workload coverage by
\[
R_{\Pi;K,\delta}
=
\int
\ind[e_K(u)\leq\delta]\,d\Pi(u).
\]
If the evaluator design gives a $\operatorname{Binomial}(M,p(u))$ marginal count for every fixed $u$ and is independent of the outer sample, the proofs apply verbatim to this coverage.
Evaluator counts may still be dependent across units through shared rows.  This extension supports a claim over a declared task workload; without a law $\Pi$, averaging unrelated benchmark problems has no population interpretation.
\end{remark}

If the evaluator catalogue itself is the deployment target, the exact hypergeometric flags $\ind[e_{M,K}(T_a)\leq\delta]$ replace interval certificates and there is no evaluator-superpopulation confidence budget.  If a deterministic evaluator design is intended to approximate a limiting $\nu_X$, a justified envelope from Proposition~\ref{prop:rate-certificate} may instead be used as the interval
\[
I_a=
[\mu_M(T_a)-b_{X,M}(T_a),\mu_M(T_a)+b_{X,M}(T_a)]\cap[0,1].
\]
The same proofs then apply with evaluator failure budget zero.  Without either stochastic sampling from $\nu_X$ or such an approximation envelope, a finite catalogue supports only census-relative claims.

\subsection{The Finite-to-Ideal Bridge}

\begin{theorem}[Fixed-Panel Finite-to-Ideal Limit]\label{thm:finite-to-ideal}
Fix $K$ and $T$.  Let an admissible sequence of censuses satisfy
\[
\mu_{X,n}(T)\longrightarrow\mu_X(T)
\qquad\text{and}\qquad
\gamma_X(T)>0.
\]
Then, as $M_n\to\infty$,
\[
D_{X,M_n}(T)\longrightarrow D_{X,\infty}(T)
\quad\text{and}\quad
e_{M_n,K}(T)\longrightarrow e_{X,K}(T).
\]
\end{theorem}

\begin{proof}
The first conclusion is Theorem~\ref{thm:deterministic-consistency}.  Write $C_n=M_n\mu_{X,n}(T)$.  For each fixed $h\in\{0,\ldots,K\}$, the hypergeometric mass can be written using falling factorials as
\[
\Prob(H_{M_n,K}=h)
=
\binom Kh
\frac{(C_n)_h(M_n-C_n)_{K-h}}{(M_n)_K}.
\]
Because $K$ is fixed and $C_n/M_n\to\mu_X(T)$, this converges to
\[
\binom Kh\mu_X(T)^h(1-\mu_X(T))^{K-h},
\]
with the usual endpoint conventions.  The finite census target equals the limiting target for all sufficiently large $n$, and summing the relevant finite set of masses gives the asserted convergence of error probabilities.
\end{proof}

The theorem answers the fixed-panel question for values such as $K=10,50,100$: under the declared census convergence assumption, exact without-replacement errors approach the ideal binomial error.  Convergence by itself does not provide a numerical rate.  The following certificate shows what additional information is sufficient.

\begin{proposition}[Explicit Finite-to-Ideal Error Envelope]\label{prop:finite-to-ideal-rate}
Fix $K$ and $T$, and suppose
\[
|\mu_{X,n}(T)-\mu_X(T)|\leq b_{X,n}(T)
\quad\text{and}\quad
|\mu_{X,n}(T)-\tau|>b_{X,n}(T).
\]
If
\[
K\mu_{X,n}(T)(1-\mu_{X,n}(T))\geq1,
\]
then
\[
|e_{M_n,K}(T)-e_{X,K}(T)|
\leq
\frac{K-1}{M_n-1}+K b_{X,n}(T).
\]
Without the displayed variance condition, the universally valid but usually looser replacement for the first term is
\[
1-\frac{(M_n)_K}{M_n^K},
\]
the probability of at least one repeated index in $K$ draws with replacement.
\end{proposition}

\begin{proof}
Proposition~\ref{prop:rate-certificate} makes the finite and limiting target decisions identical.  At the empirical parameter $p_n=\mu_{X,n}(T)$, Ehm's hypergeometric-to-binomial total-variation bound gives the sampling-without-replacement application
\[
\dTV\!\left(
\operatorname{Hypergeometric}(M_n,M_np_n,K),
\operatorname{Binomial}(K,p_n)
\right)
\leq \frac{K-1}{M_n-1}
\]
under the stated variance condition~\cite{ehm1991binomial}.  For the unconditional part, couple $K$ Bernoulli$(p_n)$ variables coordinatewise with $K$ Bernoulli$(\mu_X(T))$ variables.  The two count distributions differ with probability at most
\[
1-(1-|p_n-\mu_X(T)|)^K\leq Kb_{X,n}(T).
\]
The triangle inequality gives the result.  For the universal alternative, couple $K$ uniform index draws with replacement to a uniform ordered sample without replacement; the constructions agree whenever the with-replacement indices are all distinct.
\end{proof}

For $M=1500$, whenever the displayed variance condition holds, the Ehm term is approximately $0.0060$, $0.0327$, and $0.0660$ for $K=10,50,100$, respectively.  These are worst-case distributional envelopes at the same empirical mean, not the actual tail errors, which remain exactly computable.  The second term shows why a quantitative finite-to-infinite claim still needs a justified bound on the census approximation $b_{X,n}(T)$.

\begin{corollary}[Finite Experimental Coverage]\label{cor:finite-coverage}
Fix resolutions $T_1,\ldots,T_A$ with $\gamma_X(T_a)>0$ and $e_{X,K}(T_a)\neq\delta$ for every $a$.  Along an admissible evaluator-census sequence,
\[
R_{K,\delta}^{(A,M_n)}
\longrightarrow
\frac1A\sum_{a=1}^A\ind[e_{X,K}(T_a)\leq\delta].
\]
\end{corollary}

\begin{proof}
Theorem~\ref{thm:finite-to-ideal} gives convergence of every one of the finitely many error values.  The assumption that no limiting error equals the certification boundary makes each indicator eventually constant.  Average the indicators.
\end{proof}

If $a_1,\ldots,a_A$ are themselves drawn i.i.d.\ from $\pi$ and the ideal errors are available, the right-hand empirical coverage is an average of i.i.d.\ binary variables with mean $R_{X;K,\delta}$.  Hoeffding therefore gives deviation probability at most $2e^{-2A\varepsilon^2}$.  When ideal errors are replaced by finite-census estimates, Corollary~\ref{cor:finite-coverage} identifies the additional approximation layer and its boundary condition.

Figure~\ref{fig:resolvability} illustrates the resulting pointwise-to-family
workflow on synthetic values.

\begin{figure}[ht]
\centering
\includegraphics[alt={Synthetic resolvability dashboard showing task clarity, exact minimum panel sizes, cumulative coverage, and simultaneous lower coverage bounds.},width=\textwidth]{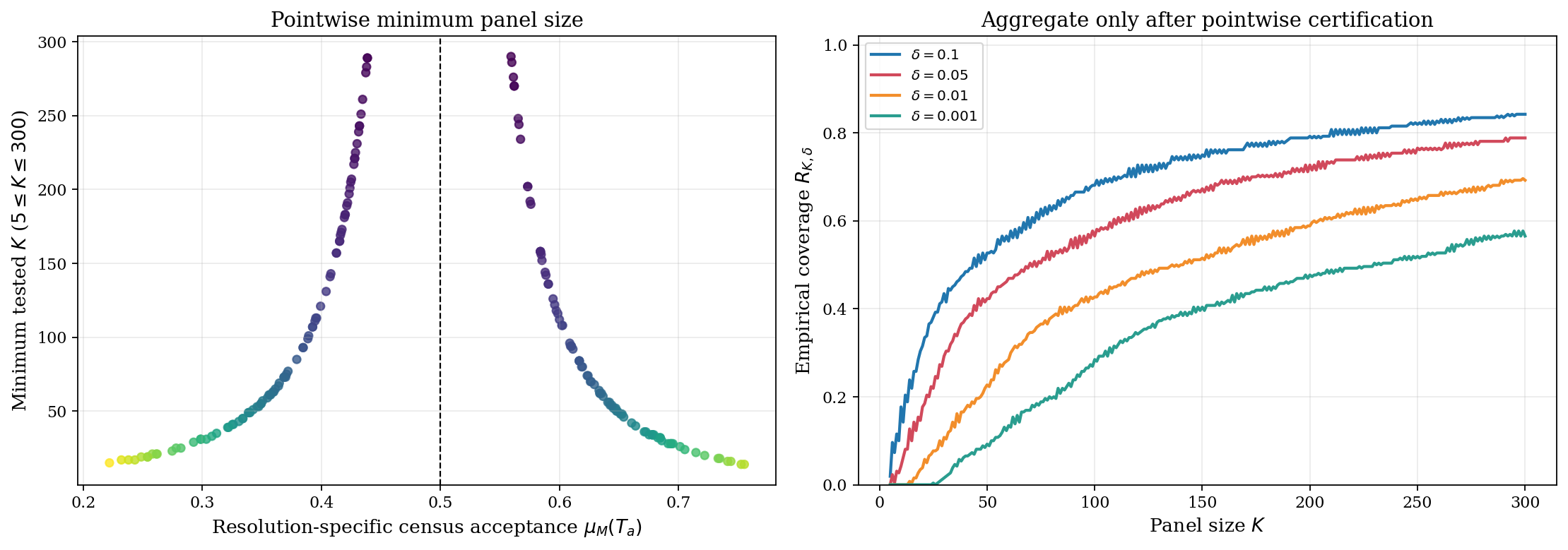}
\caption{Illustrative empirical resolvability analysis.  Each resolution has its own census acceptance rate and exact minimum panel size.  Family-level coverage is obtained only after the pointwise errors have been calculated.  The synthetic values illustrate the protocol and are not empirical claims about a deployed LLM population.}
\label{fig:resolvability}
\end{figure}

\paragraph{Robustness to evaluator population choice.}
The composition of a real LLM population is contestable.  Let
\[
\mathcal N_X=\{\nu_X^{(1)},\ldots,\nu_X^{(L)}\}
\]
be a predeclared collection of plausible evaluator weightings or finite census designs.  A conservative ADS may require the $(K,\delta,\beta)$ criterion to hold for every $\nu_X^{(\ell)}$, and may separately require that the limiting decision itself be invariant across the collection.  Failure of decision invariance is reported as population dependence, not hidden by averaging the populations.
The calibration data must support each claimed law.  In particular, rows drawn
from one $\nu_X^{(0)}$ do not automatically have binomial marginals under
$\nu_X^{(\ell)}$.  Plug-in reweighting of the same matrix is a useful
descriptive sensitivity check, but it is not a certificate for an alternative
population unless the design supplies valid stratified, importance-weighted, or
new-sample inference.

\section{Non-Adaptive Corruption as Population Contamination}\label{sec:corruption}

The adversarial model is placed at the sampling interface.  We begin with an honest frozen census and allow an adversary to replace a bounded number of its eligible entries before the uniform random ordering is drawn.  The adversary may know the problem, resolution, honest census, global rule, and threshold, but it does not observe or influence the subsequent random permutation.  This is a non-adaptive contamination model in the sense of robust statistics~\cite{huber1964robust,huber2009robust}.

\begin{definition}[Non-Adaptive Census Corruption]\label{def:nonadaptive-corruption}
Let $\mathbf y_M\in\{0,1\}^M$ be the honest global-verdict census and let $b\in\{0,1,\ldots,M\}$ be an integer budget.  A $b$-corruption is a fixed vector $\widetilde{\mathbf y}_M\in\{0,1\}^M$ chosen before $\Pi_M$ such that
\[
d_H(\mathbf y_M,\widetilde{\mathbf y}_M)\leq b,
\]
where $d_H$ is Hamming distance.  Write
\[
\widetilde C_M=\sum_i\widetilde y_{M,i},
\qquad
\widetilde\mu_M=\widetilde C_M/M.
\]
\end{definition}

The corruption may represent replaced evaluators, falsified global votes, or an altered eligible population.  It does not model an adversary that sees the sampled panel and then selects whom to corrupt; that adaptive problem has a different probability space.

\subsection{Census-Level Robustness}

\begin{proposition}[Clarity Is a Finite-Population Robustness Radius]\label{prop:finite-robustness}
Every $b$-corruption satisfies
\[
|\widetilde\mu_M-\mu_M|\leq\frac bM.
\]
If
\[
\frac bM<\gamma_M=|\mu_M-\tau|,
\]
then the corrupted and honest census decisions agree:
\[
\widetilde D_M=D_M.
\]
\end{proposition}

\begin{proof}
Changing one binary coordinate changes the total by at most one, so
$|\widetilde C_M-C_M|\leq b$.  Division by $M$ gives the first claim.  A perturbation strictly smaller than the distance from $\mu_M$ to $\tau$ cannot cross the threshold.
\end{proof}

The strict inequality treats both sides of the inclusive threshold uniformly.  On the accepting side, equality may still preserve acceptance; on the rejecting side it may reach the inclusive boundary and flip the decision.

\subsection{Exact Worst-Case Panel Error}

Conditional on any fixed corrupted vector, Theorem~\ref{thm:hypergeometric} still applies with $C_M$ replaced by $\widetilde C_M$.  The following result eliminates the need to enumerate every corruption pattern.

\begin{theorem}[Exact Non-Adaptive Worst Case]\label{thm:corrupt-exact}
Let $q_K=\lceil\tau K\rceil$ and let the error target be the \emph{honest} census decision $D_M$.

If $D_M=1$, the greatest rejection probability among all $b$-corruptions is
\[
\Prob\!\left(
H^-_{M,K}<q_K
\right),
\qquad
H^-_{M,K}\sim
\operatorname{Hypergeometric}\bigl(M,(C_M-b)_+,K\bigr).
\]
If $D_M=0$, the greatest acceptance probability is
\[
\Prob\!\left(
H^+_{M,K}\geq q_K
\right),
\qquad
H^+_{M,K}\sim
\operatorname{Hypergeometric}\bigl(M,\min(M,C_M+b),K\bigr).
\]
\end{theorem}

\begin{proof}
For an honest accepting census, rejection is monotone decreasing in the number of positive entries.  An adversary therefore changes as many positive entries to zero as the budget permits, giving $(C_M-b)_+$.  Hypergeometric distributions are stochastically increasing in their number of population successes, which can be seen by coupling populations that differ by one zero-to-one replacement under the same sampled index set.  The rejecting case is symmetric: maximizing positives gives $\min(M,C_M+b)$ and maximizes the upper tail.
\end{proof}

Thus the robust finite-panel certificate is still an exact hypergeometric tail.  It depends on the honest pointwise clarity through $C_M$ and on the integer corruption budget $b$; it does not require defining an evaluator's ``accuracy'' relative to the very population decision being estimated.

\begin{corollary}[Concentration Under Corruption]\label{cor:corrupt-concentration}
Suppose $b/M<\gamma_M$ and let
\[
\gamma_M^{\mathrm{rob}}=\gamma_M-b/M>0.
\]
Then the worst-case non-adaptive panel error is at most
\[
\exp\!\left(
-\frac{2K(\gamma_M^{\mathrm{rob}})^2}{1-(K-1)/M}
\right).
\]
\end{corollary}

\begin{proof}
Every admissible corrupted mean remains on the honest side of the threshold with distance at least $\gamma_M^{\mathrm{rob}}$.  Apply Theorem~\ref{thm:serfling-panel} to the corrupted finite population.
\end{proof}

\subsection{Asymptotic Contamination}

Let a sequence of honest censuses satisfy $\mu_M\to\mu$ with $\gamma=|\mu-\tau|>0$, and let $b_M/M\to\alpha$.  If $\alpha<\gamma$, Proposition~\ref{prop:finite-robustness} implies that every sufficiently large corrupted census has the honest limiting decision.  If the corrupted proportions themselves converge to $\widetilde\mu$, their centered random-order paths satisfy Theorem~\ref{thm:bridge} about $\widetilde\mu$ whenever $\widetilde\mu\in(0,1)$.  Relative to the honest threshold, contamination appears as an additional deterministic drift bounded in magnitude by $\alpha$.

At $\alpha\geq\gamma$, no population-independent guarantee is possible: an admissible contamination can move the mean to or across the threshold.  This limitation is deterministic and precedes any CLT.

\subsection{Corruption of a Selected Panel}

For completeness, suppose a panel is first sampled honestly and at most $f$ of its reported votes are then changed.  Pathwise, the corrupted positive count $\widetilde H$ obeys
\[
H-f\leq\widetilde H\leq H+f.
\]
Consequently, when the target decision is acceptance, the worst-case failure event is contained in
\[
\{H<q_K+f\},
\]
and when the target is rejection it is contained in
\[
\{H\geq q_K-f\}.
\]
These shifted exact tails are useful for a fixed panel budget, but they are not the same model as non-adaptive corruption of the eligible population.  A deployment must state which mechanism is in scope.

Figure~\ref{fig:corruption} shows how the exact finite-census tail changes with
the pre-sampling corruption budget.

\begin{figure}[ht]
\centering
\includegraphics[alt={Worst-case panel disagreement probability versus corruption fraction for several panel sizes and honest clarity levels.},width=\textwidth]{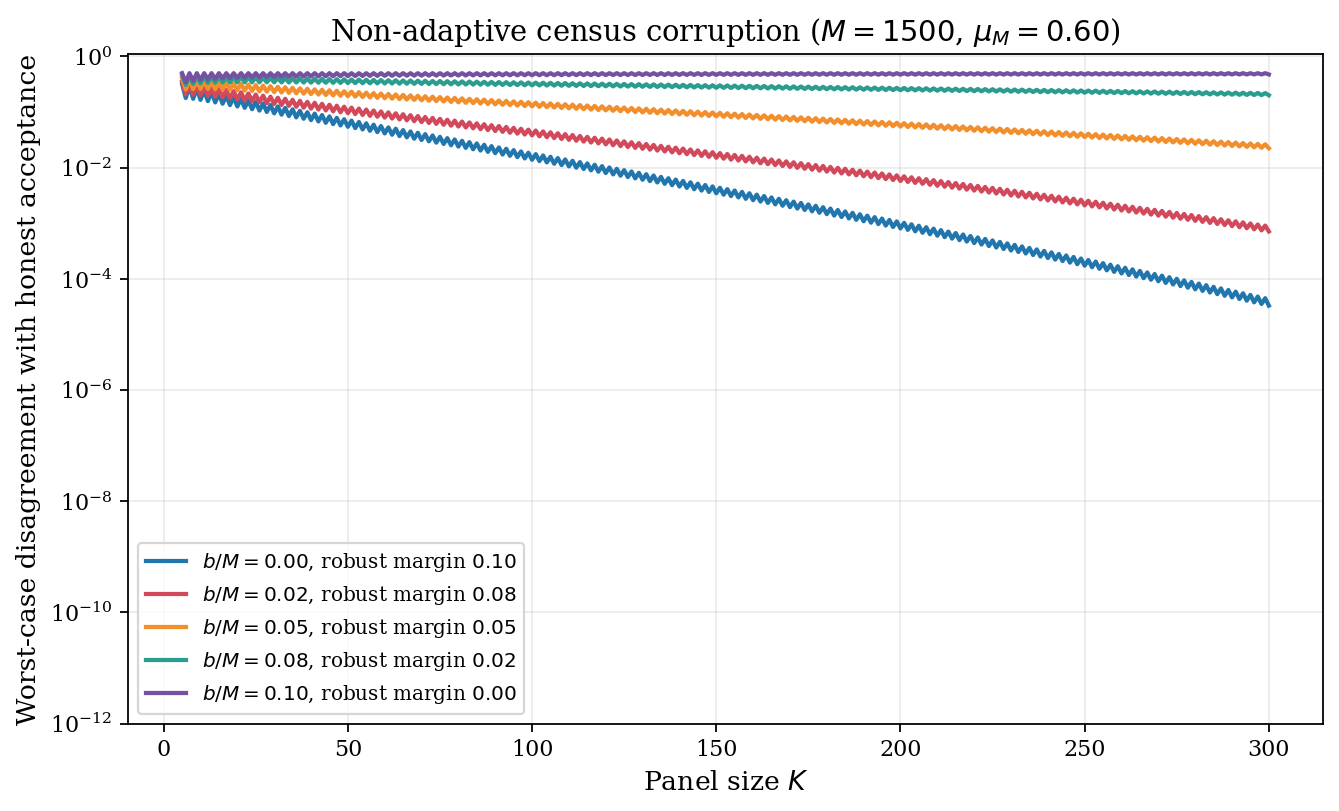}
\caption{Exact worst-case panel disagreement under non-adaptive census
corruption.  Corruption reduces the effective clarity by at most $b/M$.  At
equality the uniform no-flip guarantee ends; once $b/M$ exceeds the honest
clarity, the population decision can cross the threshold.  Under the inclusive
rule, equality can still preserve an accepting decision.}
\label{fig:corruption}
\end{figure}

\subsection{From a Corruption Fraction to a Sensitivity Profile}

The question ``what fraction of the network must an adversary control?'' has no single answer at the statistical decision layer.  Even after the selection mechanism and attack objective have been fixed, outcome manipulation depends on the task clarity and on the panel size.  Moreover, controlling a strict majority of sampled identities and moving a semantically ambiguous decision across its threshold are different events.

The distinction is transparent in the large-population i.i.d.\ benchmark.  Specialize to $\tau=1/2$ and odd $K$, and write $q_K=(K+1)/2$.  If each sampled identity is adversarial independently with probability $\alpha$, the strict-majority committee-capture curve is
\[
P_K^{\mathrm{cap}}(\alpha)
=
\Prob\!\left\{
\operatorname{Binomial}(K,\alpha)\geq q_K
\right\}.
\]
This curve concerns committee membership only and is independent of task clarity.

For outcome manipulation, two attack models must be distinguished.  First, suppose $\alpha$ literally denotes a fixed Byzantine share of a large network.  Each independently sampled identity is Byzantine with probability $\alpha$ and then votes against the honest decision; conditional on being honest, it supports that decision with probability $1/2+\gamma_H$.  The unconditional honest-side vote probability is
\[
p_{\mathrm{net}}(\alpha,\gamma_H)
=
(1-\alpha)\left(\frac12+\gamma_H\right),
\]
so the fixed-share network attack curve is
\begin{equation}\label{eq:network-share-attack-curve}
P_K^{\mathrm{net}}(\alpha;\gamma_H)
=
\Prob\!\left\{
\operatorname{Binomial}\!\left(
K,(1-\alpha)\left[\frac12+\gamma_H\right]
\right)<q_K
\right\}.
\end{equation}
Here $\gamma_H$ is clarity within the honest subpopulation.  This additional mixture assumption is required to interpret $\alpha$ as network ownership.  Its population-level $50\%$ frontier is the curve
\begin{equation}\label{eq:network-share-boundary}
\gamma_H=\frac{\alpha}{2(1-\alpha)},
\end{equation}
not the diagonal $\gamma_H=\alpha$.

By contrast, Definition~\ref{def:nonadaptive-corruption} permits targeted Hamming contamination: the adversary may spend its budget specifically on entries that support the honest population decision.  Orienting the binary vote toward that decision gives honest mean $1/2+\gamma$, and an $\alpha$-contamination may reduce this mean by the full $\alpha$.  Define the targeted-contamination attack curve
\begin{equation}\label{eq:semantic-attack-curve}
P_K^{\mathrm{tar}}(\alpha;\gamma)
=
\Prob\!\left\{
\operatorname{Binomial}\!\left(K,
\left[\frac12+\gamma-\alpha\right]_{[0,1]}
\right)<q_K
\right\},
\end{equation}
where $[x]_{[0,1]}=\min\{1,\max\{0,x\}\}$.  At the deterministic boundary $\alpha=\gamma$, every odd panel has $P_K^{\mathrm{tar}}(\gamma;\gamma)=1/2$.  Below that boundary larger panels suppress sampling error; above it they increasingly concentrate on the adversarially shifted population decision.  Thus a larger panel amplifies whichever side of the population threshold remains after contamination.

For the same numerical $\alpha$ and clarity, targeted contamination is at least as damaging as fixed network share because
\[
\frac12+\gamma-\alpha
\leq
(1-\alpha)\left(\frac12+\gamma\right)
\qquad(0\leq\gamma\leq1/2).
\]
The distinction is operational: the left side allows the attacker to select supportive entries to replace, whereas the right side samples a pre-existing Byzantine fraction without conditioning on the honest entries' latent votes.

For a target pointwise error $\delta<1/2$, let
\[
r_{K,\delta}
=
\inf\left\{
r\in[0,1/2]:
\Prob\!\left\{
\operatorname{Binomial}(K,1/2+r)<q_K
\right\}
\leq\delta
\right\}.
\]
Binomial stochastic monotonicity gives the two operational conditions
\begin{align}
P_K^{\mathrm{net}}(\alpha;\gamma_H)\leq\delta
&\quad\Longleftrightarrow\quad
\gamma_H\geq\frac{\alpha/2+r_{K,\delta}}{1-\alpha},
\label{eq:network-share-operational-condition}\\
P_K^{\mathrm{tar}}(\alpha;\gamma)\leq\delta
&\quad\Longleftrightarrow\quad
\gamma\geq\alpha+r_{K,\delta}.
\label{eq:targeted-operational-condition}
\end{align}
Consequently, if $\Pi$ is a declared workload law over problems and resolutions, its targeted-contamination sensitivity profile is
\begin{equation}\label{eq:workload-security-profile}
S_{K,\delta}(\alpha)
=
\Prob_{u\sim\Pi}\!\left\{
\gamma(u)\geq\alpha+r_{K,\delta}
\right\}.
\end{equation}
It reports the fraction of the declared workload whose pointwise attack probability is at most $\delta$ at contamination level $\alpha$.  A scalar tolerance is a service-level slice of this curve, for example
\[
\alpha^*_{K,\delta,\beta}
=
\sup\{\alpha:S_{K,\delta}(\alpha)\geq1-\beta\}.
\]

The same nested construction used for resolvability yields a finite-sample
lower confidence bound for this profile, rather than only a plug-in sensitivity
curve.

\begin{corollary}[Targeted-Contamination Profile Certificate]
\label{cor:contamination-profile-certificate}
Fix $\delta\in(0,1/2)$ and a non-empty finite grid, declared before seeing the
matrix,
\[
\mathcal G_{\mathrm{adv}}
\subset
\{(K,\alpha):K\in\{1,3,5,\ldots\},\ \alpha\in[0,1/2]\}.
\]
Under the workload-level sampling design of
Remark~\ref{rem:workload-extension} and the assumptions of
Theorem~\ref{thm:mass-operational-certificate}, construct the mass-controlled
intervals $I_a^{\mathrm{mass}}=[L_a,U_a]$ at pointwise failure level
$\eta_E\xi_E$.  For $(K,\alpha)\in\mathcal G_{\mathrm{adv}}$, define
\[
\operatorname{Cert}^{\mathrm{tar}}_{K,\delta,\alpha}([L,U])
=
\ind\!\left[
L\geq\frac12+\alpha+r_{K,\delta}
\quad\text{or}\quad
U\leq\frac12-\alpha-r_{K,\delta}
\right],
\]
\[
\widehat S^{\mathrm{mass}}_{K,\delta}(\alpha)
=
\frac1A\sum_{a=1}^A
\operatorname{Cert}^{\mathrm{tar}}_{K,\delta,\alpha}
\bigl(I_a^{\mathrm{mass}}\bigr),
\]
and
\[
\underline S_{K,\delta}(\alpha)
=
\max\!\left\{
0,
\ell_A\!\left(
A\widehat S^{\mathrm{mass}}_{K,\delta}(\alpha);
\frac{\eta_G}{|\mathcal G_{\mathrm{adv}}|}
\right)
-\xi_E
\right\}.
\]
Then
\[
\Prob\!\left(
S_{K,\delta}(\alpha)
\geq
\underline S_{K,\delta}(\alpha)
\text{ for every }(K,\alpha)\in\mathcal G_{\mathrm{adv}}
\right)
\geq1-\eta_E-\eta_G.
\]
Consequently, any pair selected from the same grid whose lower bound is at
least $1-\beta$ certifies the corresponding targeted-contamination workload
profile at the stated confidence level.
\end{corollary}

\begin{proof}
For a fixed workload unit $u$, if the interval $[L,U]$ contains its honest
acceptance mean $p(u)$ and the displayed certificate equals one, then
$|p(u)-1/2|\geq\alpha+r_{K,\delta}$.  In the lower branch, reorienting the
vote toward the honest rejecting decision replaces $p(u)$ by $1-p(u)$, whose
mean is at least $1/2+\alpha+r_{K,\delta}$.  By
Equation~\eqref{eq:targeted-operational-condition}, the pointwise
targeted-contamination attack probability is therefore at most $\delta$.  Thus
the interval certificate is a one-sided lower bound on the latent good-set
indicator defining $S_{K,\delta}(\alpha)$.

Let $F(\mathbf Z)$ be the generator mass of configurations whose evaluator
interval misses its honest mean, as in the proof of
Theorem~\ref{thm:mass-operational-certificate}.  Tonelli's theorem and Markov's
inequality give $\Prob\{F(\mathbf Z)>\xi_E\}\leq\eta_E$.  Conditional on the
shared evaluator rows $\mathbf Z$, the observed targeted-certificate indicators
are i.i.d.\ Bernoulli over generator draws.  Exact one-sided binomial inversion
and a union bound over $\mathcal G_{\mathrm{adv}}$ lower-bound their conditional
population masses simultaneously with failure probability at most $\eta_G$.
On $F(\mathbf Z)\leq\xi_E$, each such certificate mass is at most the
corresponding true profile plus $\xi_E$.  Rearranging and intersecting the two
events gives the result.
\end{proof}

Under the fixed-share network model, the corresponding workload profile replaces the event in~\eqref{eq:workload-security-profile} by
\[
\left\{
\gamma_H(u)\geq
\frac{\alpha/2+r_{K,\delta}}{1-\alpha}
\right\}.
\]
An analogous certificate applies to the fixed-share profile when the
calibration rows sample the declared honest-subpopulation law needed to identify
$\gamma_H$.  Without that design, synthetic simulation can validate the
calculation but cannot identify a deployed curve.

Figure~\ref{fig:adversarial-security-profiles} evaluates these quantities on the nominal FeeManager size grid,
\[
5,7,11,13,23,25,47,49,95,97,191,193,383,385,767,769,1535,1537,
\]
as read from \href{https://github.com/genlayerlabs/genlayer-consensus/blob/8795606edaa84ac8b46bc6687d9a4d5ad9d25f57/contracts/FeeManager.sol#L53-L71}{\texttt{FeeManager.sol}} at archived revision \texttt{8795606e}, which is also pinned by the companion TLA repository's \texttt{refs/consensus} submodule.  Each value is evaluated here as an independent strict-majority panel.  In the pinned contracts, appeal feasibility depends on the round index and on the currently available unconsumed validators; successor construction can instead combine prior committees or snapshot the active set.  The companion executable specification also applies a minus-two sizing rule after consecutive unsuccessful appeals.  Thus neither the number of realizable appeals nor this complete nominal grid is a path-independent deployed schedule.  The normal/appeal pairs are shown with solid/dashed lines.  The final panel plots
\[
\alpha^*_{K,\delta}(\gamma)
=
\gamma-r_{K,\delta}
\]
when the right-hand side is non-negative; below $r_{K,\delta}$ even the honest panel cannot meet the declared error target.

\begin{figure}[htbp]
\centering
\includegraphics[alt={Four non-adaptive corruption sensitivity charts over the nominal FeeManager size grid: committee capture, targeted manipulation at two clarity levels, and maximum tolerable contamination.},width=\textwidth]{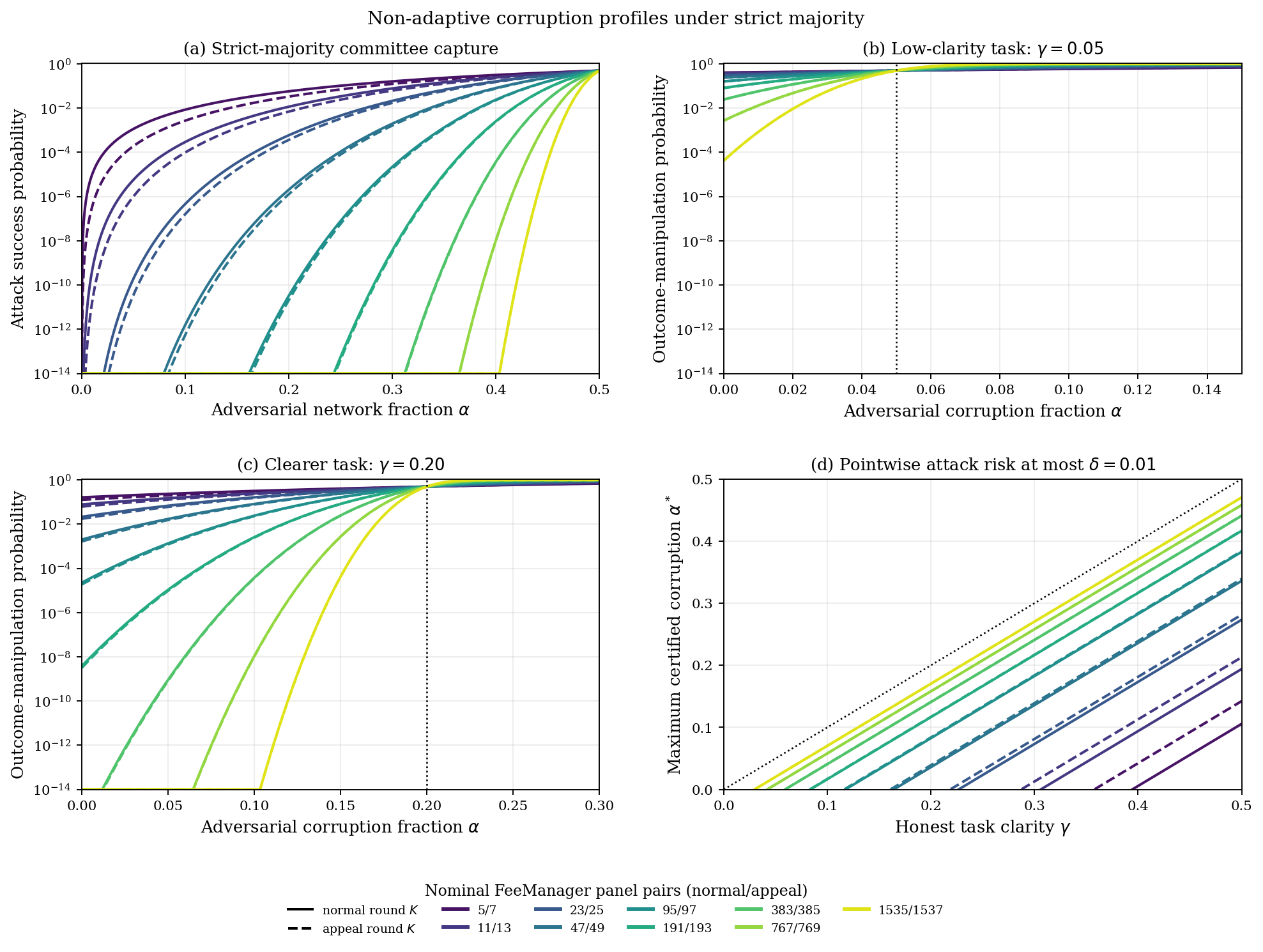}
\caption{Non-adaptive sensitivity over the nominal FeeManager size grid under the large-population i.i.d.\ benchmark.  Panel (a) shows committee capture; panels (b)--(c) show targeted contamination at two clarity levels; panel (d) gives the maximum contamination compatible with pointwise error $\delta=0.01$.  The grid does not model deployed validator selection and is not a deployed-network security guarantee.}
\label{fig:adversarial-security-profiles}
\end{figure}

For the direct network-control reading, hold the fixed Byzantine share constant.  Figure~\ref{fig:panel-clarity-adversary} places the nominal panel size on the horizontal axis and the minimum clarity within the honest subpopulation on the vertical axis.  Each curve is one Byzantine network share:
\[
\gamma_{H,\min}^{\mathrm{net}}(K;\alpha,\delta)
=
\frac{\alpha/2+r_{K,\delta}}{1-\alpha}.
\]
A task--panel pair on or above a curve meets the declared pointwise attack-risk target under this benchmark; a point below it does not.  Values above $1/2$ are impossible for a binary threshold and therefore identify combinations that no panel of that size can certify.

The complementary phase diagram fixes the largest nominal sensitivity size, $K=1537$, and varies both quantities continuously.  Figure~\ref{fig:attack-phase-diagram} places the two threat models side by side with common axes and color scale.  In the fixed-share network model the $50\%$ frontier is the curved boundary~\eqref{eq:network-share-boundary}, and the low-risk and high-success frontiers are
\[
\gamma_H
=
\frac{\alpha/2\mathbin{\pm}r_{K,\delta}}{1-\alpha}.
\]
Under targeted contamination the corresponding frontiers are $\gamma=\alpha$ and $\gamma=\alpha\mathbin{\pm}r_{K,\delta}$.  The visual difference is the price of allowing the attacker to choose supportive census entries rather than merely owning a pre-existing random share of identities.

The corresponding curves and phase diagrams appear in
Figures~\ref{fig:panel-clarity-adversary} and~\ref{fig:attack-phase-diagram}.

\clearpage
\renewcommand{\headeright}{Supplementary Material}
\pdfbookmark[0]{Supplementary Material}{part1-supplementary-material}
\documentdivision{Supplementary Material}
\setcounter{section}{0}
\setcounter{figure}{0}
\setcounter{table}{0}
\setcounter{equation}{0}
\renewcommand{\thesection}{S\arabic{section}}
\renewcommand{\thefigure}{S\arabic{figure}}
\renewcommand{\thetable}{S\arabic{table}}
\renewcommand{\theequation}{S\arabic{equation}}
\renewcommand{\theHsection}{supplement.\arabic{section}}
\renewcommand{\theHfigure}{supplement.\arabic{figure}}
\renewcommand{\theHtable}{supplement.\arabic{table}}
\renewcommand{\theHequation}{supplement.\arabic{equation}}
\section{Functional Limits and Late Reversals}\label{supp:functional-limits}

The joint variable $\mathbf Y=(Y_1,Y_2,\ldots)$ is infinite-dimensional, but the endpoint $K^{-1}\sum_{i=1}^K Y_i$ is scalar.  To retain the whole evolution as the panel grows, we map the binary sequence to its interpolated partial-sum path.

\subsection{The i.i.d.\ Superpopulation Path}

Fix $(X,T)$ and suppose $Y_1,Y_2,\ldots$ are the coordinate votes of Proposition~\ref{prop:joint-law}, with $0<\mu<1$ and $\sigma^2=\mu(1-\mu)$.  For $t\in[0,1]$, set $m=\lfloor nt\rfloor$ and define the polygonally interpolated process
\[
W_n(t)
=
\frac{1}{\sigma\sqrt n}
\left(
\sum_{i=1}^{m}(Y_i-\mu)
+(nt-m)(Y_{m+1}-\mu)
\right),
\]
where the interpolation term is understood as zero at $t=1$.

\begin{theorem}[Superpopulation Functional CLT]\label{thm:donsker}
As random elements of $C[0,1]$ with the uniform topology,
\[
W_n\Rightarrow W,
\]
where $W$ is standard Brownian motion with covariance
\[
\Cov(W(s),W(t))=\min(s,t).
\]
\end{theorem}

\begin{proof}
The centered Bernoulli increments are i.i.d., have mean zero and finite non-zero variance, and are bounded.  The polygonal form of Donsker's invariance principle therefore applies~\cite{billingsley1999convergence}.
\end{proof}

At $t=1$, Theorem~\ref{thm:donsker} reduces to the ordinary scalar CLT.  Its additional content is joint weak convergence of the complete path, which permits continuous path functionals such as maxima over a fixed interval to be analyzed.

\subsection{The Without-Replacement Census Path}

Now let $\mathbf y_M$ be a deterministic binary census, let $\Pi_M$ be a uniform random permutation, and suppose
\[
\mu_M\longrightarrow\mu\in(0,1).
\]
For $t\in[0,1]$, $m=\lfloor Mt\rfloor$, define
\[
B_M(t)
=
\frac{1}{\sqrt{M\mu_M(1-\mu_M)}}
\left(
\sum_{i=1}^{m}(y_{M,\Pi_M(i)}-\mu_M)
+(Mt-m)(y_{M,\Pi_M(m+1)}-\mu_M)
\right),
\]
again taking the interpolation term to be zero at $t=1$.

Unlike the i.i.d.\ path,
\[
B_M(1)=0
\]
identically: after the entire census has been revealed, its centered total is known.

\begin{theorem}[Finite-Population Functional CLT]\label{thm:bridge}
As $M\to\infty$,
\[
B_M\Rightarrow B^\circ
\qquad\text{in }C[0,1],
\]
where $B^\circ$ is a standard Brownian bridge with covariance
\[
\Cov(B^\circ(s),B^\circ(t))
=
\min(s,t)-st.
\]
\end{theorem}

\begin{proof}
The randomly permuted centered census is a triangular array of exchangeable, bounded increments with zero total.  Its total squared norm is $M\mu_M(1-\mu_M)$, which diverges because $\mu_M\to\mu\in(0,1)$, while its maximum increment is at most one.  Thus the maximal-increment (Lindeberg) ratio tends to zero.  The functional combinatorial central limit theorem for uniform random permutations applies~\cite{barbour2009functional}.  The endpoint constraint gives the tied-down Gaussian limit.  Its covariance can also be read from Proposition~\ref{prop:bridge-covariance} below.
\end{proof}

\begin{proposition}[Pre-limit Covariance]\label{prop:bridge-covariance}
Let
\[
R_{M,k}=\sum_{i=1}^k(y_{M,\Pi_M(i)}-\mu_M).
\]
For $1\leq k\leq\ell\leq M$,
\[
\Cov(R_{M,k},R_{M,\ell})
=
\mu_M(1-\mu_M)\frac{k(M-\ell)}{M-1}.
\]
Consequently, if $k/M\to s$ and $\ell/M\to t$ with $s\leq t$, the covariance after the normalization of Theorem~\ref{thm:bridge} tends to $s(1-t)$.
\end{proposition}

\begin{proof}
Every permuted coordinate has variance $\mu_M(1-\mu_M)$.  Two distinct coordinates have covariance $-\mu_M(1-\mu_M)/(M-1)$ because their sum is fixed.  Expanding the covariance of the two overlapping partial sums gives
\[
\mu_M(1-\mu_M)
\left(k-\frac{k\ell-k}{M-1}\right)
=
\mu_M(1-\mu_M)\frac{k(M-\ell)}{M-1}.
\]
\end{proof}

The Brownian bridge is therefore the Gaussian limit induced by sampling without replacement and the known terminal census total.

Figure~\ref{fig:nested-paths} illustrates the unpinned and pinned endpoint
geometries; the proofs do not rely on the displayed draws.

\begin{figure}[ht]
\centering
\includegraphics[alt={Two nested-panel fluctuation paths: an unpinned partial-sum path and a random-order census path pinned to zero at the full census.},width=\textwidth]{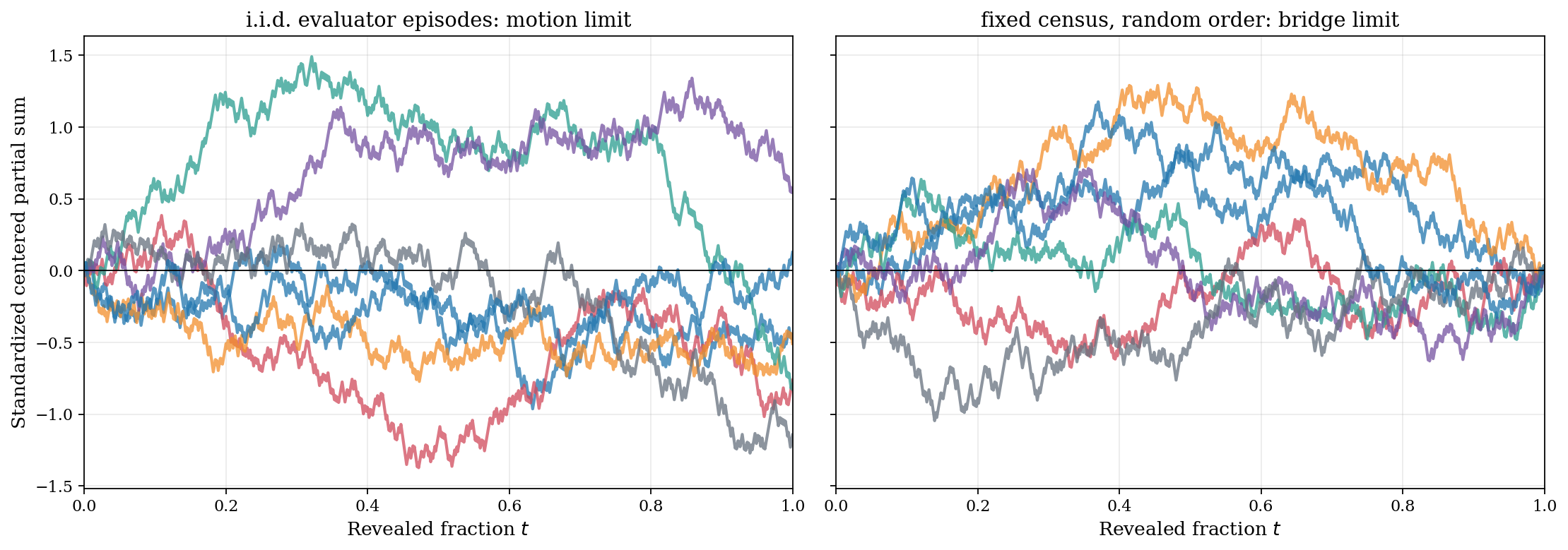}
\caption{Nested-panel fluctuations.  Left: centered i.i.d.\ Bernoulli partial sums, whose functional limit is Brownian motion.  Right: random-order paths through a fixed binary census, pinned to zero at the full census and converging to a Brownian bridge.  The plotted paths are illustrative draws, not evidence for the theorems.}
\label{fig:nested-paths}
\end{figure}

\subsection{Threshold Drift and Local Ambiguity}

The decision compares the cumulative count with the line $k\tau$.  In the i.i.d.\ regime,
\[
\frac{1}{\sigma\sqrt n}
\sum_{i=1}^{\lfloor nt\rfloor}(Y_i-\tau)
=
W_n(t)
+
\frac{\lfloor nt\rfloor(\mu-\tau)}{\sigma\sqrt n}
+O(n^{-1/2}).
\]
For fixed non-zero clarity, the deterministic drift grows as $\sqrt n$ and eventually dominates the $O_\Prob(1)$ fluctuation.  The non-degenerate asymptotic regime near ambiguity occurs when the population mean approaches the threshold at the $n^{-1/2}$ scale.

\begin{theorem}[Local-Clarity Bridge Limit]\label{thm:local-bridge}
Let $\tau\in(0,1)$ and let binary censuses satisfy
\[
\mu_M
=
\tau+\frac{c}{\sqrt M}+o(M^{-1/2})
\]
for some $c\in\mathbb R$.  For $t\in[k/M,(k+1)/M]$ with $k=0,\ldots,M-1$, define the threshold-centered interpolated path by
\[
Z_M(t)
=
\frac{1}{\sqrt{M\tau(1-\tau)}}
\left(
\sum_{i=1}^{k}(y_{M,\Pi_M(i)}-\tau)
+(Mt-k)(y_{M,\Pi_M(k+1)}-\tau)
\right).
\]
At $t=1$, use the full sum over $i=1,\ldots,M$.
Then
\[
Z_M\Rightarrow
B^\circ(t)+\frac{c}{\sqrt{\tau(1-\tau)}}t
\qquad\text{in }C[0,1].
\]
\end{theorem}

\begin{proof}
Add and subtract $\mu_M$.  The centered term equals $B_M(t)$ multiplied by
\[
\sqrt{\frac{\mu_M(1-\mu_M)}{\tau(1-\tau)}}\longrightarrow1.
\]
Uniformly in $t$, the remaining deterministic term converges to
$ct/\sqrt{\tau(1-\tau)}$.  Apply Theorem~\ref{thm:bridge} and Slutsky's theorem in $C[0,1]$.
\end{proof}

Theorem~\ref{thm:local-bridge} supplies an approximation for path events away from the singular initial time.  The next corollary evaluates one such event in closed form.

\begin{corollary}[Late-Escalation Stability Profile]\label{cor:late-escalation-profile}
Under the assumptions of Theorem~\ref{thm:local-bridge}, suppose $c\neq0$ and fix $a\in(0,1)$.  Let
\[
\mathcal S_{M,a}
=
\left\{
D_{M,k}=D_M
\text{ for every }k=\lceil aM\rceil,\ldots,M
\right\}.
\]
Writing $\Phi$ for the standard normal distribution function and
\[
\lambda=\frac{|c|}{\sqrt{\tau(1-\tau)}},
\]
the probability that no later nested panel reverses the full-census decision satisfies
\[
\Prob(\mathcal S_{M,a})
\longrightarrow
2\Phi\!\left(\lambda\sqrt{\frac{a}{1-a}}\right)-1.
\]
Consequently, the limiting probability of at least one disagreement after fraction $a$ is
\[
2\Phi\!\left(-\lambda\sqrt{\frac{a}{1-a}}\right).
\]
For a target late-disagreement probability $\varepsilon\in(0,1)$, the limiting stability probability is at least $1-\varepsilon$ exactly when
\[
a
\geq
\frac{z_{1-\varepsilon/2}^2}
{\lambda^2+z_{1-\varepsilon/2}^2},
\qquad
z_u=\Phi^{-1}(u).
\]
\end{corollary}

\begin{proof}
At a grid point, $Z_M(k/M)\geq0$ is equivalent to
$H_{M,k}\geq k\tau$, hence to $D_{M,k}=1$ under the inclusive threshold
convention.  If $c>0$, then $D_M=1$ for all sufficiently large $M$, and
$\mathcal S_{M,a}$ is the event that the interpolated path remains
non-negative on $[a_M,1]$, where $a_M=\lceil aM\rceil/M$.  If $c<0$, then
$D_M=0$ eventually and the corresponding event for $-Z_M$ uses a strict
positive inequality.  This finite-$M$ distinction is the parity/tie effect;
it disappears in the limit below.

For deterministic $x_M\to x$ uniformly and $a_M\to a$, uniform continuity
of $x$ gives
\[
\inf_{a_M\leq t\leq1}x_M(t)
\longrightarrow
\inf_{a\leq t\leq1}x(t).
\]
Thus the moving left endpoint does not change the continuous-mapping
argument.  Moreover, the conditional reflection formula used below gives a
continuous distribution for the minimum of a Brownian bridge with positive
endpoints; since $X(a)$ has a density, the limiting minimum has no atom at
zero.  The weak and strict finite-$M$ inequalities therefore have the same
limit.  Theorem~\ref{thm:local-bridge} and the extended continuous mapping
theorem reduce both cases to
\[
\Prob\!\left(
\inf_{a\leq t\leq1}
\{B^\circ(t)+\lambda t\}>0
\right).
\]

Set $X(t)=B^\circ(t)+\lambda t$.  Then $X(1)=\lambda$ and
$X(a)\sim\mathcal N(\lambda a,a(1-a))$.  Conditional on $X(a)=x>0$, the segment on $[a,1]$ is a Brownian bridge from $x$ to $\lambda$.  The reflection principle gives its probability of remaining positive as
\[
1-\exp\!\left(-\frac{2x\lambda}{1-a}\right).
\]
Indeed, reflection at the first hit of zero replaces the Brownian transition density over duration $1-a$ from $p_{1-a}(\lambda-x)$ by $p_{1-a}(\lambda+x)$; their ratio is $\exp\{-2x\lambda/(1-a)\}$.  Dividing the killed transition density by the unrestricted one gives the displayed conditional probability.
It is zero for $x\leq0$.  With $d=\lambda\sqrt{a/(1-a)}$, integration over $X(a)$ gives
\[
\Prob(X(a)>0)
-
\E\!\left[
e^{-2\lambda X(a)/(1-a)}\ind[X(a)>0]
\right]
=
\Phi(d)-\Phi(-d)
=
2\Phi(d)-1.
\]
The disagreement probability is its complement.  Solving
$2\Phi(d)-1\geq1-\varepsilon$ for $a$ gives the final expression.
\end{proof}

For a general interval $0<a<b<1$, write
$X(t)=B^\circ(t)+ct/\sqrt{\tau(1-\tau)}$.  Conditioning on its two Gaussian endpoints gives the computable boundary-crossing formula
\[
\Prob\!\left(\inf_{a\leq t\leq b}X(t)>0\right)
=
\E\!\left[
\ind[X(a)>0,X(b)>0]
\left\{1-\exp\!\left(-\frac{2X(a)X(b)}{b-a}\right)\right\}
\right].
\]
At finite $M$, exact dynamic programming should still be preferred when this path event itself is used for certification.

\subsection{Stabilization}

\begin{theorem}[Almost-Sure Superpopulation Stabilization]\label{thm:iid-stabilization}
Let $Y_i$ be i.i.d.\ Bernoulli$(\mu)$ and define
\[
D_K=\ind\!\left[\frac1K\sum_{i=1}^K Y_i\geq\tau\right].
\]
If $\mu\neq\tau$, then
\[
D_K\longrightarrow\ind[\mu\geq\tau]
\qquad\text{almost surely},
\]
and there is an almost surely finite random $K_0$ after which every decision is equal to the limit decision.  If $\mu=\tau\in(0,1)$, the decisions do not converge almost surely.
\end{theorem}

\begin{proof}
For $\mu\neq\tau$, the strong law gives $K^{-1}\sum_{i=1}^K Y_i\to\mu$ almost surely, so the averages are eventually on the same side of $\tau$ as $\mu$.  At $\mu=\tau$, the centered non-degenerate Bernoulli random walk has positive and negative excursions infinitely often almost surely (for example by the law of the iterated logarithm), so the thresholded sequence cannot stabilize.
\end{proof}

\begin{theorem}[Late-Census Stabilization]\label{thm:census-stabilization}
Suppose $\mu_M\to\mu\neq\tau$.  For every fixed $a\in(0,1)$,
\[
\Prob\!\left(
D_{M,k}=D_M\text{ for every }k=\lceil aM\rceil,\ldots,M
\right)
\longrightarrow1.
\]
\end{theorem}

\begin{proof}
For sufficiently large $M$, $\gamma_M\geq|\mu-\tau|/2$ and $D_M=\ind[\mu\geq\tau]$.  By Theorem~\ref{thm:serfling-panel} and a union bound, the probability of any disagreement for $k\geq aM$ is at most
\[
M\exp\!\left(-2aM\left(\frac{|\mu-\tau|}{2}\right)^2\right),
\]
which tends to zero.
\end{proof}

\begin{remark}[Choice of path space]\label{rem:path-space}
The raw state space $\{0,1\}^{\Naturals}$ with its product sigma-algebra is an infinite-dimensional measurable product, but it is not a vector space.  Embedding it in $\ell^\infty$ does not make a Banach-valued CLT automatic: independent non-degenerate Gaussian coordinates are not bounded almost surely, so the putative limit need not be an $\ell^\infty$-valued random element.  General Banach-space CLTs require geometric and tightness hypotheses beyond a formal second moment~\cite{araujo1980central,ledoux1991probability}.

The natural construction here is instead the partial-sum map into a path space.  Theorems~\ref{thm:donsker} and~\ref{thm:bridge} are functional CLTs for that system-relevant random function.  If each evaluator itself returned a countably infinite vector of component verdicts, a separate Hilbert- or Banach-valued model and a meaningful norm would have to be specified.
\end{remark}

\section{Supplementary Empirical Protocol}\label{supp:llm-protocol}

A task-specific estimate of $\nu_X$ can be built from real LLM configurations only after the inferential target and sampling design have been declared.  A confirmatory campaign should satisfy four requirements.
\begin{enumerate}
    \item \emph{Declare the target.}  Freeze $X$, $G$, $\tau$, the tie rule, prompts, tools, and time window.  For multiple problems, specify a workload--generator law $\Pi$ rather than treating a benchmark list as a population.  State whether each catalogue is a finite target, an i.i.d.\ sample from a named law, or a deterministic approximation with a justified envelope.
    \item \emph{Freeze the design.}  Predeclare generator and evaluator catalogues, weights, seeds, self-evaluation policy, and any nested evaluator ordering.  Generate and freeze every $T_a$ before cross-evaluation.
    \item \emph{Predeclare inference.}  Fix $(\delta,\beta)$, $(\eta_E,\eta_G)$, any mass allowance $\xi_E$, and the candidate set $\mathcal K$.  Construct the complete matrix $\mathbf Y$, compute exact census-relative errors and the applicable population intervals, and select $K$ only through a simultaneous rule stated above.
    \item \emph{Report scope and sensitivity.}  Publish the binary matrix and configuration metadata; report pointwise clarity, changes across nested designs, alternative population weights, failures, and temporal replications.  If no candidate certifies, report that outcome rather than changing the grid or budgets retrospectively.
\end{enumerate}

Stabilization across the observed sizes is evidence for a limit, not a proof of a universal population.  Theorem~\ref{thm:finite-to-ideal} is conditional on convergence and does not manufacture that premise from the observations.  A probabilistic confidence statement additionally requires the configurations to be sampled according to the declared $\nu_X$; a deterministic catalogue requires a design-based or approximation argument.

Each evaluator configuration may also generate its own resolution $T_a$.  Generation and evaluation remain distinct roles: the former induces the distribution of columns, the latter the acceptance distribution within each column.  The empirical question is therefore not one mean $\mu$ for the examination, but the distribution of the function $T\mapsto\mu_X(T)$ over generated resolutions.

\section{LLM Study Details}\label{supp:real-llm-pilot}

We exercised the single-verdict ($N=1$) computational path on a hash-pinned diagnostic pilot.  The catalogue consisted of $A=50$ frozen equivalence-principle decisions divided into four construction strata: clear accept, clear reject, benign ambiguity, and input-presentation stress.  The last stratum modifies the evidence presentation; it is not Byzantine evaluator corruption and is separate from the attack models of Section~\ref{sec:corruption}.  These units form a stratified seed catalogue, not an i.i.d.\ sample from an identified workload--generator law.  Their accept, reject, and construction-underdetermined labels were assigned during construction and were not independently adjudicated external truth.  Because each prompt requested one global ACCEPT/REJECT value, this pilot does not test a nontrivial multicomponent map $G$.

The evaluator design drew $M=40$ rows i.i.d.\ with replacement from an equal-weight catalogue of 21 routed evaluator configurations.  These are router-level configurations, not claims of 21 statistically independent training lineages.  The shared-row theorem permits arbitrary dependence among the 50 cells of a row but requires rows to be i.i.d.\ from the declared law.  Distinct seeds do not rule out persistent dependence induced by shared random provider state, routing incidents, tools, or cached context, none of which can be audited fully from the archived binary ledger.  The realized sample contained 18 distinct configurations; repeated configurations retained distinct row and cell seeds.  Every row evaluated every frozen unit, producing an $M\times A=40\times50$ matrix.  The analysis fixed
\[
\tau=0.5,\qquad
\delta=0.01,\qquad
\beta=0.40,\qquad
\eta_E=\eta_G=0.025,\qquad
\xi_E=0.05,
\]
and the GenLayer panel grid
\[
\mathcal K=
\{5,7,11,13,23,25,47,49,95,97,191,193,383,385,767,769,1535,1537\}.
\]
This is a predeclared statistical grid based on the nominal FeeManager sizes,
not a claim that every value is reachable as a deployed protocol panel; the
capacity endgame has separate validator-eligibility and fixed-threshold
semantics.
Here $M$ is the calibration-row count, whereas $K$ is the size of a future i.i.d.\ panel from the declared superpopulation.  Thus candidates with $K>M$ are meaningful in this model; the restriction $K\leq M$ applies only to a without-replacement panel drawn from the realized finite census.
Malformed or exhausted transport responses were mapped to rejection under the frozen failure rule.

The family weights below were included in the frozen input dataset before
evaluator calls and were not fitted to the vote matrix.  The dataset SHA-256 is
\[
\texttt{fe5c11531def6849c439dc60919361b60141fee26fc427eeaf5159a0e1b2ea66}.
\]

\begin{table}[htbp]
\centering
\footnotesize
\begin{tabular}{@{}lrrr@{}}
\toprule
Family $f$ & Units $n_f$ & Family mass $w_f$ & Unit mass $r_i=w_f/n_f$ \\
\midrule
\texttt{sla\_enforcement}              & 4 & 0.13 & 0.0325 \\
\texttt{payment\_dispute}              & 4 & 0.13 & 0.0325 \\
\texttt{marketplace\_dispute}          & 4 & 0.12 & 0.0300 \\
\texttt{prediction\_market\_objective} & 4 & 0.10 & 0.0250 \\
\texttt{parametric\_insurance}         & 4 & 0.10 & 0.0250 \\
\texttt{airdrop\_task\_eval}           & 4 & 0.08 & 0.0200 \\
\texttt{prediction\_market\_subjective}& 4 & 0.07 & 0.0175 \\
\texttt{grant\_scoring}                & 4 & 0.06 & 0.0150 \\
\texttt{bug\_bounty\_severity}         & 4 & 0.06 & 0.0150 \\
\texttt{content\_moderation}           & 4 & 0.06 & 0.0150 \\
\texttt{sports\_referee}               & 4 & 0.04 & 0.0100 \\
\texttt{model\_fingerprinting}         & 3 & 0.03 & 0.0100 \\
\texttt{emergency\_halt}               & 3 & 0.02 & 0.0067 \\
\midrule
Total                                  & 50 & 1.00 & --- \\
\bottomrule
\end{tabular}
\caption{Declared nonuniform law on the finite unit catalogue.  The
\texttt{family\_weight} values encode a suggested GenLayer-like workload prior,
not estimated deployment frequencies.  The uniform-catalogue result is
reported alongside this design-weighted result.  Unit masses are displayed to
four decimal places; all calculations use $w_f/n_f$ without rounding.}
\label{tab:catalogue-weights}
\end{table}

\paragraph{Observed matrix.}
Of the $2{,}000$ cells, $1{,}983$ returned valid structured responses, 15 were malformed, and two exhausted transport retries.  The empirical 40-row majority decision matched the construction label on all 37 construction-resolvable units; the individual-cell agreement on those units was $0.994$.  For acceptance proportion $p$, binary entropy means
\[
h_2(p)=-p\log_2p-(1-p)\log_2(1-p),
\]
with the usual zero convention and units of bits.  Its mean was $0.041$ on construction-resolvable units and $0.333$ on construction-unresolvable units.

The observed margins were highly saturated: 44 of 50 units had at most one dissenting vote, including all 12 units in the input-presentation-stress stratum.  Only the benign-ambiguity stratum supplied appreciable mass near the aggregation threshold.  The pilot therefore verifies the mechanics of the single-verdict code path, but it is not a high-power stress test for intermediate clarity, a multicomponent $G$, or resistance to hostile presentation.

\paragraph{Separated certificates.}
At $K=5$ and $K=7$, no unit passed the evaluator-superpopulation pointwise
interval rule under either the familywise or mass-controlled construction.
This is a power statement about the chosen $(M,\delta)$ design,
not evidence that those panel sizes have error one or that a deployed first
round is insecure.

For the fully enumerated 50-unit catalogue, familywise intervals at level
$\eta_E/A=0.0005$ certify 45 units at $K=47$.  Corollary~\ref{cor:finite-generator-catalogue}
then gives lower coverage $45/50=0.900$ for uniform weights and $0.8775$ for the
declared catalogue weights, with confidence $1-\eta_E=0.975$ under the declared
i.i.d.\ evaluator-row model.  The mass-controlled finite-catalogue alternative
certifies 46 units and gives declared-weight lower coverage
$0.8925-0.05=0.8425$.

For the frozen 40-row census, exact without-replacement tails require no
confidence interval.  Forty-seven units satisfy
$K_{\mathrm{stable}}\leq7$ at $\delta=0.01$; \texttt{mkt-03},
\texttt{ins-03}, and \texttt{air-03} require 26, 34, and 39, respectively.
This statement is conditional on the realized verdict census, and $K\leq40$.
Putting the two results side by side, the same matrix has 47 of 50
frozen-census stability results at $K=7$ and zero evaluator-superpopulation
interval certificates at $K=7$.  The gap reflects different reference objects
and different uncertainty, not conflicting calculations.

Finally, the hypothetical outer-sampling calculation uses the
mass-controlled intervals.  At $K=47$, 46 of 50 units certify; exact one-sided
binomial inversion gives
$\ell_{50}(46;0.025/18)-0.05=0.6912$, while the Hoeffding form gives
$46/50-0.05-\sqrt{\log(18/0.025)/100}=0.6135$.  At $K=23$, 44 units certify and
the exact form gives $\ell_{50}(44;0.025/18)-0.05=0.6372$, hence crosses target
$0.60$ there.  These would be simultaneous 95\% lower bounds under the
hypothetical i.i.d.\ outer design; the factor 18 retains the complete
predeclared grid.  The original pilot analysis used Hoeffding and first crossed
at $K=47$.  Since
the outer catalogue was not sampled i.i.d.\ from a declared law $\Pi$, none of
these outer lower-bound values has the workload-population confidence
interpretation of Theorem~\ref{thm:mass-operational-certificate}.  We retain
$K=47$ as the audit reference rather than reselecting it retrospectively.

The pilot value $\beta=0.40$ was a feasibility choice, not a judgment that
$40\%$ unresolved workload mass is acceptable for deployment.  The audit value
$K=47$ is not a production panel-size recommendation.
Even under the hypothetical i.i.d.\ outer design, certification at these
values would give only the Corollary~\ref{cor:deployment-error} deployment-error
bound $(1-0.40)0.01+0.40=0.406$; because the actual outer design does not meet
that corollary's sampling premise, $0.406$ is not a pilot guarantee.

\paragraph{Frozen-census contamination sensitivity.}
Applying Theorem~\ref{thm:corrupt-exact} column by column gives the exact
worst-direction result when at most $b$ of the 40 eligible verdicts are changed
before the random panel is drawn.  Table~\ref{tab:pilot-census-corruption}
reports a compact slice of \path{finite_census_security.csv}.  It is conditional
on this frozen census and targets its honest decision; it does not model
adaptive corruption after panel selection, validator identities, or network
ownership.

\begin{table}[htbp]
\centering
\small
\begin{tabular}{@{}rrrrrr@{}}
\toprule
Budget $b$ & $b/M$ & $K=5$ & $K=7$ & $K=11$ & $K=39$ \\
\midrule
0 & 0.000 & 47 & 47 & 47 & 50 \\
1 & 0.025 & 46 & 47 & 47 & 49 \\
2 & 0.050 & 45 & 46 & 47 & 49 \\
4 & 0.100 & 44 & 44 & 47 & 48 \\
\bottomrule
\end{tabular}
\caption{Number of the 50 catalogue units whose exact worst-case
without-replacement panel error remains at most $\delta=0.01$ under
non-adaptive census contamination.}
\label{tab:pilot-census-corruption}
\end{table}

\begin{table}[htbp]
\centering
\small
\begin{tabular}{@{}lr@{}}
\toprule
Metric & Observed result \\
\midrule
Frozen workload units & 50 \\
Evaluator rows / total calls & 40 / 2,000 \\
Realized / declared routed configurations & 18 / 21 \\
Valid structured responses & 1,983 / 2,000 (99.15\%) \\
Construction-resolvable majority matches & 37 / 37 \\
Audit reference & $K=47$ \\
Fixed-catalogue familywise certificate & 45 / 50 (90.0\%) \\
Declared-weight catalogue lower coverage & 0.8775 (97.5\% confidence) \\
Frozen-census units stable by $K=7$ & 47 / 50 \\
Outer diagnostic, exact / Hoeffding & 0.6912 / 0.6135 \\
Construction-underdetermined units passing mass-controlled pointwise rule & 9 / 13 \\
\bottomrule
\end{tabular}
\caption{Detailed hash-pinned LLM pilot summary.  Run
\texttt{pilot-21f-final-20260803-01}; immutable input and derived-artifact
hashes are recorded in the provenance manifest.}
\label{tab:real-llm-pilot}
\end{table}

Table~\ref{tab:real-llm-pilot} records the complete numerical summary;
Figure~\ref{fig:real-llm-pilot} shows the diagnostic curves in the main text.

\paragraph{Sensitivity and interpretation.}
The mass-controlled intervals for \texttt{mkt-03}, \texttt{ins-03}, and
\texttt{air-03} crossed $\tau=0.5$.  Equal weighting of the 18 realized
configurations changed no majority decision; leave-one-configuration
recalculations changed only \texttt{air-03}.  Treating every non-success as
acceptance likewise changed only \texttt{air-03}; that policy and a
successes-only calculation changed the pointwise certification status of
\texttt{bug-03}.  These are descriptive sensitivity checks, not certificates
for alternative evaluator populations or missingness mechanisms.

Of the 13 units constructed to be underdetermined, nine passed the
mass-controlled pointwise rule at $K=47$ and ten did so at $K\geq95$.
Construction status is not
independent semantic adjudication, so this does not establish an empirical
relation between true semantic underdetermination and reproducibility.
Operational resolvability concerns agreement with a declared population
decision and does not, by itself, identify semantic truth.

\section{Synthetic Design Details}\label{supp:synthetic-role}

Before calling real models, the complete experiment can be simulated from declared latent acceptance probabilities.  Such a simulation can verify the implementation of quotas, exact tails, interval inversion, simultaneous lower bounds, and data-dependent selection; it can also estimate the power and cost of candidate values of $(A,M,\mathcal K,\delta,\beta,\eta_E,\eta_G,\xi_E)$.  Row-level latent variables can induce strong dependence across columns while keeping rows i.i.d., directly stress-testing the shared-evaluator design allowed by Theorems~\ref{thm:operational-certificate} and~\ref{thm:mass-operational-certificate}.

A reproducible calibration study uses latent acceptance means
\[
(0.30,0.40,0.46,0.49,0.51,0.54,0.60,0.70)
\]
with generator weights
\[
(0.15,0.20,0.05,0.10,0.10,0.05,0.20,0.15).
\]
It fixes $\tau=0.5$, $\delta=0.05$, $\beta=0.40$, $\eta_E=\eta_G=0.025$, $\xi_E=0.05$, and
\[
\mathcal K=\{21,41,61,81,101,151,201,301\}.
\]
In the dependent regime, each evaluator row uses a common uniform variable across all columns with probability $\rho=0.9$ and independent uniforms otherwise.  Thus every column count retains its exact $\operatorname{Binomial}(M,p_a)$ marginal while the columns can be strongly dependent.  Table~\ref{tab:operational-monte-carlo} and Figure~\ref{fig:operational-monte-carlo} report $2{,}000$ repetitions for each design and dependence regime.

The exact coverage is $0.30$ for $K\leq61$ on the candidate grid and $0.70$ for $K\geq81$.  Nevertheless, the baseline design usually cannot certify the $0.60$ target: interval uncertainty near the pointwise certification boundary and the finite-sample lower bound reduce its power.  Increasing both sampling dimensions makes the same target detectable.  The experiment therefore illustrates why failure to certify is inconclusive and why $(A,M)$ should be chosen by power analysis before an expensive evaluation campaign.

Simulation cannot identify the generator distribution of $T\mapsto\mu_X(T)$ for real models, validate the declared LLM population weights, detect provider or temporal drift, or measure external semantic truth.  Those are empirical properties, not consequences of the probability model.  Synthetic studies can therefore replace real LLM calls for theorem calibration and power analysis, but not for a claim that an ADS is useful on a real task population.

\paragraph{External validation.}
If independent ground-truth labels are available, one may additionally measure false acceptance and false rejection relative to them.  Those semantic metrics are separate from $e_{M,K}$, which measures coordination with a population decision.  Neither high clarity nor high panel--census agreement proves external correctness.

\section{One Seed Supplies Countably Many Private Draws}\label{supp:seed}

The evaluator definition uses one seed while allowing multiple component queries and private tool observations.  The following standard construction makes that convention explicit.

\begin{lemma}[Digit Splitting]\label{lem:digit-splitting}
There is a Borel map
\[
s:[0,1]\longrightarrow[0,1]^{\Naturals},
\qquad
s(\omega)=(\omega^{(1)},\omega^{(2)},\ldots),
\]
such that, for $\omega\sim\operatorname{Uniform}[0,1]$, the coordinates $\omega^{(j)}$ are i.i.d.\ uniform.
\end{lemma}

\begin{proof}
Choose for each $\omega\in[0,1)$ the binary expansion that is not eventually all ones and define the dyadic-rational exceptions arbitrarily.  Under Lebesgue measure, the binary digits are i.i.d.\ Bernoulli$(1/2)$.  Partition their countable index set into countably many disjoint infinite subsets using a bijection $p:\Naturals^2\to\Naturals$.  For each $j$, use the digits with indices $p(j,1),p(j,2),\ldots$ as the binary expansion of $\omega^{(j)}$.  Disjoint digit families are independent, each coordinate is uniform, and every coordinate map is a pointwise limit of measurable finite sums~\cite{billingsley1995probability}.
\end{proof}

This lemma concerns private randomness.  It does not transform a common random internet state into independent evidence; a common state must appear as a shared coordinate in the probability space.

\section{Exact Computation and Reproducibility}\label{supp:computation}

For each resolution column, an implementation needs only the tuple
\[
(M,C_M,K,\tau)
\]
to compute the panel--census error.  With $q_K=\lceil\tau K\rceil$, use a numerically stable hypergeometric survival function or cumulative distribution function rather than evaluating large binomial coefficients directly:
\[
e_{M,K}=
\begin{cases}
F_{\mathrm{HG}}(q_K-1;M,C_M,K),&D_M=1,\\
1-F_{\mathrm{HG}}(q_K-1;M,C_M,K),&D_M=0.
\end{cases}
\]
If $\tau$ is specified as a finite decimal or rational, compute $q_K$ with exact rational arithmetic; a floating-point product that lies just above an integer can otherwise increase the inclusive quota by one.

For the ideal-population certificate in Theorem~\ref{thm:operational-certificate}, a two-sided Clopper--Pearson interval with miscoverage $\alpha$ can be computed from beta quantiles as
\[
L_{M,\alpha}(c)
=
\begin{cases}
0,&c=0,\\
F^{-1}_{\operatorname{Beta}(c,M-c+1)}(\alpha/2),&c>0,
\end{cases}
\]
and
\[
U_{M,\alpha}(c)
=
\begin{cases}
1,&c=M,\\
F^{-1}_{\operatorname{Beta}(c+1,M-c)}(1-\alpha/2),&c<M.
\end{cases}
\]
For the familywise theorem, set $\alpha=\eta_E/A$ and denote the resulting
interval by $I_a^{\mathrm{cert}}$.  For the mass-controlled theorem, set
$\alpha=\eta_E\xi_E$, denote the interval by $I_a^{\mathrm{mass}}$, and charge
$\xi_E$ in the final coverage bound.  Compute each family of interval endpoints
once and then, for every $K\in\mathcal K$, apply
$\operatorname{Cert}_{K,\delta}$.  The same intervals serve all candidate values
of $K$; no additional evaluator-level union bound over $\mathcal K$ is needed.
The generator layer uses $|\mathcal K|$ because the latent resolvability
indicators can change with $K$.

For $\star\in\{\mathrm{cert},\mathrm{mass}\}$, let
\[
S_K^\star
=
\sum_{a=1}^A\operatorname{Cert}_{K,\delta}(I_a^\star),
\qquad
\widehat R_{K,\delta}^\star=S_K^\star/A,
\]
and define the one-sided outer binomial limit
\[
L_K^{G,\star}=
\begin{cases}
0,&S_K^\star=0,\\
F^{-1}_{\operatorname{Beta}(S_K^\star,A-S_K^\star+1)}
  \!\left(\eta_G/|\mathcal K|\right),&S_K^\star>0.
\end{cases}
\]
An executable reference implementation should return the two distinct exact
bounds
\[
\underline R_{K,\delta}^{\mathrm{cert}}
=L_K^{G,\mathrm{cert}},
\qquad
\underline R_{K,\delta}^{\mathrm{mass}}
=\max\{0,L_K^{G,\mathrm{mass}}-\xi_E\},
\]
and the corresponding Hoeffding benchmarks
\[
\max\!\left\{0,
\widehat R_{K,\delta}^{\mathrm{cert}}
-\sqrt{\frac{\log(|\mathcal K|/\eta_G)}{2A}}
\right\},
\qquad
\max\!\left\{0,
\widehat R_{K,\delta}^{\mathrm{mass}}-\xi_E
-\sqrt{\frac{\log(|\mathcal K|/\eta_G)}{2A}}
\right\}.
\]
Returning the complete grid rather than only the first passing $K$ preserves
parity effects and makes a failed certificate diagnostically useful.

A reproducible generator--evaluator record should include at least:
\begin{itemize}[nosep]
    \item hashes of $X$, $T_a$, prompts, $G$, and tool policies;
    \item model providers, versions, decoding parameters, seeds, and timestamps;
    \item component vector $\mathbf V_{b,a}$ and global value $Y_{b,a}$;
    \item generator and evaluator weights;
    \item eligibility and self-evaluation policy;
    \item the precommitted evaluator ordering used for nested-census diagnostics;
    \item every alternative population weighting included in robustness analysis.
\end{itemize}

Exact probabilities should be used whenever the sampling design is uniform without replacement.  Monte Carlo is appropriate only for a different design whose law lacks a practical closed form, or for visualizing path functionals not used as formal certificates.

The reference script \texttt{operational\_monte\_carlo.py} performs a separate calibration and power study under a known discrete generator population.  It mixes shared-row and idiosyncratic uniform variables so that evaluator columns can be strongly dependent while every column retains its exact binomial marginal.  A simulation violation occurs if any reported lower bound exceeds the analytically known coverage at any candidate $K$.  Figure~\ref{fig:operational-monte-carlo} and the CSV summaries in \path{theory/part1/results/} are reproduced from the companion release by the Makefile target \texttt{theory-monte-carlo}.  This is a regression check on the implementation and a study-design tool, not a numerical proof of the theorem or evidence about real LLM clarity.

The Makefile target \texttt{theory-security-curves} reproduces
Figures~\ref{fig:attack-phase-diagram}
and~\ref{fig:adversarial-security-profiles} and their exact binomial margin floors.  The
deterministic source and CSV output are stored under \path{theory/part1/}.  The
panel ladder and complete numerical environment are pinned in the script and
\path{flake.lock}.  No Monte Carlo
approximation is used in these curves.

\subsection{Real-LLM Pilot Artifacts}\label{supp:llm-pilot-artifacts}

The empirical and numerical companion is archived in
\href{https://github.com/genlayerlabs/ads-reproducibility/releases/tag/v1.0.0}{release
\texttt{v1.0.0}} of the reproducibility repository, at full commit
\href{https://github.com/genlayerlabs/ads-reproducibility/commit/94967ec58e627a0c0b4d9768202723b9b63c325c}{\nolinkurl{94967ec58e627a0c0b4d9768202723b9b63c325c}}.
The release includes a license, citation metadata, a dependency lock, and a
pinned numerical environment.

The bundle under \path{results/paper/} corresponds to run
\path{pilot-21f-final-20260803-01}.  Its provenance record distinguishes the
original execution and analysis revisions from the later offline binary-vote
reanalysis, retains full identifiers, and records the manifest, dataset, prompt,
dependency-lock, raw-ledger, and derived-artifact hashes.

The release bundle contains:
\begin{itemize}[nosep]
    \item \path{global_verdicts.csv}, the sanitized $2{,}000$-cell binary ledger;
    \item \path{evaluator_rows.csv}, the realized 40-row configuration design;
    \item \path{per_item.csv}, exact counts, intervals, clarity, entropy, and certification flags;
    \item \path{finite_catalogue_coverage.csv},
    \path{finite_census_per_item.csv}, and
    \path{finite_census_coverage.csv}, the fixed-reference results;
    \item \path{operational_coverage.csv}, the exact and Hoeffding
    outer-sampling diagnostics;
    \item \path{population_sensitivity.csv},
    \path{failure_policy_sensitivity.csv}, and
    \path{finite_census_security.csv}, the declared descriptive sensitivities;
    \item \path{security_profiles.csv}, the plug-in asymptotic attack diagnostic;
    \item \path{provider_summary.csv}, serving-model, error, and latency aggregates; and
    \item \path{summary.json} and \path{provenance.json}, the claim scope and hash manifest.
\end{itemize}
Raw provider text and routing traces are not published in Git; their immutable ledger hashes remain in the provenance record.  The sanitized binary ledger suffices to recompute the vote-based results in Table~\ref{tab:real-llm-pilot} and Figure~\ref{fig:real-llm-pilot}, but an external reader cannot independently reparse the original text or authenticate the router-reported serving identities from this release.  Those are explicit auditability limitations of the diagnostic pilot.  No authorization material is present in the bundle.

The 15 malformed outputs were split evenly across three routes:
\texttt{claude-opus-4-8}, \texttt{deepseek-v4-flash}, and
\texttt{deepseek-v4-pro}.  One \texttt{glm-4.6} call and one
\texttt{kimi-k2.6} call exhausted transport retries.  All 17 cells followed the
frozen conservative rejection rule.  Mean recorded latency was $8.53$ seconds.

The 40 declared i.i.d.\ evaluator rows realized 18 of the 21 router configurations.  The absent routes were \texttt{gpt-5.6-luna}, \texttt{claude-sonnet-4-6}, and \texttt{muse-spark-1.1}; this is a possible outcome under categorical sampling with replacement, not a routing substitution.  A separate diagnostic load gate exercised all 21 routes before the primary run, but its calls are not included in the $40\times50$ matrix.

The companion bundle retains \path{security_profiles.csv} as a reproducible plug-in sensitivity diagnostic.  We do not promote that diagnostic to a paper figure: it substitutes 40-row point estimates into a $K=1537$ attack law, carries neither evaluator-estimation uncertainty nor a workload-population interpretation, and is therefore too easy to mistake for a deployed-network security claim.

\subsection{Code and Data Availability}\label{supp:availability}

The manuscript's \LaTeX{} source is distributed with the arXiv submission.
The deterministic simulation code, numerical tests, figure generators,
sanitized empirical ledger, realized evaluator design, analysis code,
dependency locks, and hash manifest are distributed in the exact companion
release identified above.  The arXiv source bundle contains every table and
figure required to compile the paper and does not depend on Git submodules,
private credentials, or live model endpoints.  Provider calls themselves
require access to the named router and cannot be replayed from the archived
source alone; all theorem checks and vote-based analyses run offline.

\section{Logical Structure}\label{supp:dependency}

Table~\ref{tab:logical-structure} records the dependency of the main objects.  The finite-census and superpopulation branches share an encoded decision problem but retain distinct probability laws until a finite-to-ideal approximation is supplied.

\begin{table}[htbp]
\centering
\small
\renewcommand{\arraystretch}{1.12}
\begin{tabularx}{\textwidth}{@{}l>{\raggedright\arraybackslash}X>{\raggedright\arraybackslash}X@{}}
\toprule
Layer & Object & Output or required assumption \\
\midrule
Decision model
& $(X,T)$, evaluator episode $z$, component vector $\mathbf V$, fixed rule $G$
& Measurable global verdict $Y(z;X,T)$ \\
Finite census
& Frozen vector $\mathbf y_M$ and uniform random ordering
& Exact hypergeometric endpoints and Brownian-bridge path limit \\
Superpopulation
& I.i.d.\ episode law $\nu_X$
& Exact binomial endpoints and Brownian-motion path limit \\
Population transfer
& Admissible finite designs and, for numerical bounds, an approximation envelope
& Comparison of census-relative and population-relative errors \\
Generator family
& Resolution law $\pi$ and generator--evaluator matrix
& Familywise or mass-controlled lower confidence bound \\
Robustness
& Declared contamination or fixed-share attack model
& Pointwise and workload-level sensitivity profiles \\
\bottomrule
\end{tabularx}
\caption{Logical layers of the ADS analysis.  No finite-census result is transferred to an ideal population without an explicit convergence model.}
\label{tab:logical-structure}
\end{table}

\end{document}